\newif\iflong
\newif\ifshort

\iflong
\else
\shorttrue
\fi

\documentclass{scrartcl}

\usepackage{times}

\usepackage{url}
\usepackage[utf8]{inputenc}
\usepackage[small]{caption}
\usepackage{graphicx}
\usepackage{amsmath}
\usepackage{booktabs}

\usepackage{amsfonts,amssymb}
\usepackage{amsthm}

\usepackage{nicefrac}
\usepackage{todonotes}
\usepackage{enumitem}
\usepackage{dsfont}
\usepackage[ruled, noend]{algorithm2e}

\usepackage{natbib}

\definecolor{my-blue}{HTML}{1E88E5}
\definecolor{my-blue-light}{HTML}{64B5F6}
\definecolor{my-blue-verylight}{HTML}{B3E5FC}
\definecolor{my-blue-dark}{HTML}{0D47A1}

\usepackage[unicode=true, bookmarks=false, breaklinks=true, pdfborder={0 0 1}, colorlinks=true]{hyperref}
\hypersetup{linkcolor=my-blue-dark, citecolor=my-blue-dark, filecolor=my-blue, urlcolor=my-blue-dark}

\theoremstyle{plain}
\newtheorem{theorem}{Theorem}

\newtheorem{example}{Example}

\newtheorem{proposition}[theorem]{Proposition}
\newtheorem{definition}{Definition}

\newtheorem{claim}{Claim}

\newenvironment{claimproof}{\par\noindent\underline{Proof:}}{\leavevmode\unskip\penalty9999 \hbox{}\nobreak\hfill\quad\hbox{$\blacksquare$}}

\definecolor{MyOrange}{rgb}{1, 0.5, 0}
\definecolor{MyLightOrange}{rgb}{1, 0.9, 0.7}
\definecolor{MyGrey}{rgb}{0.3, 0.3, 0.3}
\definecolor{MyLightGrey}{rgb}{0.9, 0.9, 0.9}
\definecolor{MyGreen}{rgb}{0, 0.6, 0}
\definecolor{MyBlue}{rgb}{0, 0, 0.6}
\definecolor{MyLightBlue}{rgb}{0.7, 0.8, 1}
\definecolor{MyRed}{rgb}{0.7, 0, 0}
\definecolor{MyLightRed}{rgb}{1, 0.7, 0.7}
\definecolor{MyLightYellow}{HTML}{f9f6ec}

\usepackage{pifont}
\usepackage{amsmath}
\newcommand{\cmark}{{\color{MyGreen}\ding{51}}}

\newcommand{\xmark}{{\color{MyRed}\ding{55}}}

\newcommand{\tuple}[1]{\left\langle #1 \right\rangle}

\renewcommand{\phi}{\varphi}

\newcommand{\Rplus}{\ensuremath{\mathbb{R}_{\geq 0}}}

\DeclareMathOperator*{\argmax}{arg\,max}
\DeclareMathOperator*{\argmin}{arg\,min}

\newcommand{\agentSet}{\mathcal{N}}
\newcommand{\projSet}{\mathcal{P}}
\newcommand{\allocSet}{\mathcal{A}}

\newcommand{\profile}{\boldsymbol{A}}
\newcommand{\pbRule}{F}
\newcommand{\share}{\mathit{share}}
\newcommand{\fairshare}{\mathit{fairshare}}

\newcommand{\bigO}{\mathcal{O}}
\newcommand{\complexP}{\ensuremath{\mathtt{P}}}
\newcommand{\complexNP}{\ensuremath{\mathtt{NP}}}
\newcommand{\complexcoNP}{\ensuremath{\mathtt{coNP}}}
\newcommand{\complexFPT}{\ensuremath{\mathtt{FPT}}}

\newcommand{\rulex}{MES$_{\share{}}$}

\renewcommand{\paragraph}[1]{\medskip\noindent\textbf{#1}}

\title{Fairness in Participatory Budgeting via Equality of Resources\thanks{This is the full version of a paper published in the proceedings on the 22nd International Conference on Autonomous Agents and Multiagent Systems~\citep{MREL23aamas}.}
}

\author{
	Jan Maly$^1$ \hspace{1em}
	Simon Rey$^1$ \hspace{1em}
	Ulle Endriss$^1$ \hspace{1em}
	Martin Lackner$^2$
}

\date{
	$^1$ILLC, University of Amsterdam\\
	$^2$DBAI, TU Wien\\[1em]
	\{j.f.maly, s.j.rey, u.endriss\}@uva.nl \hspace{2em}
	lackner@dbai.tuwien.ac.at
}

\begin{document}
	
	\maketitle
	
	\begin{abstract}
		\noindent
		We introduce a family of normative principles to assess fairness in the context of participatory budgeting. These principles are based on the fundamental idea that budget allocations should be fair in terms of the resources invested into meeting the wishes of individual voters. This is in contrast to earlier proposals that are based on specific assumptions regarding the satisfaction of voters with a given budget allocation. We analyse these new principles in axiomatic, algorithmic, and experimental terms.
	\end{abstract}
	
	\section{Introduction}
	
	Budgeting, \textit{i.e.}, the allocation of money or other sparse resources to specific projects,
	is one of the key decisions any political body or organisation has to take.
	Participatory budgeting (PB) was developed in the 1990s as a method for making such decisions in
	a more democratic way, by putting it to a vote~\citep{Caba04, Shah07}. It has found rapid adoption worldwide, in particular at the municipal level~\citep{WNT21}.
	
	The most common form of eliciting the views of the voters is to ask which projects they approve of~\citep{GoelKSA19}, but the question of which voting rules should be used to select the projects to be funded is not yet settled. In this paper, we advocate for the use of fairness measures based on the 
	resources spent on behalf of individual voters to guide the search for the best and most equitable voting rules. Specifically, we focus on a measure of \emph{equality of resources}~\citep{Dworkin1, Dworkin2} called the \emph{share}, recently introduced by \citet{LMR21}. It is computed by equally dividing the cost of each funded project amongst the supporters of that project.
	
	Let us briefly motivate this approach.
	Suppose 40\% of citizens of a city support funding
	more cycling infrastructure, while 60\% are in favour of more car infrastructure.
	Then, under the kind of voting rule usually employed in practice,
	where the projects with the most support get selected, only car-centric projects would get funded.
	This clearly is not desirable.
	Instead, one would hope to select a
	\emph{proportional} outcome~\citep{FGM16, ALT18, PPS21NeurIPS, AL21, LCG22},
	funding a mixture of cycling and car infrastructure projects.
	But how should one define proportionality?
	So far, the literature has focused on generalisations
	from approval-based multiwinner voting, where we
	often aim for a proportional distribution of \emph{satisfaction} amongst voters,
	assuming that each approved candidate provides the same satisfaction
	to all of their supporters~\citep{FSST17, LaSk23}.
	However, lifting this assumption to the richer framework of PB is questionable, as projects vary in cost.
	
	So, given their approval ballot, how should one infer a voter's satisfaction for a set
	of selected projects?\footnote{In principle, there is also another possibility,
		namely to directly ask voters for their satisfaction (or utility).
		But this would imposes a significant cognitive burden on them,
		and it is debatable whether it is even possible to elicit utilities in a way
		that allows for interpersonal comparisons~\citep{Hicks,Blackorby02}.}
	Most researchers assume that the satisfaction of a voter is either equal for all
	approved projects~\citep{PPS21NeurIPS, TaFa19, LCG22}
	or proportional to the cost of a project~\citep{FGM16, ALT18, LMR21, SBY22}.
	Both assumptions are problematic.
	Regarding the former, for example, in the 2019 Toulouse participatory budgeting process
	a cycling infrastructure project costing 390,000 euros (``Tous à la Ramée à vélo! A pied, en trottinette et rollers!'') and a project
	about installing birdhouses costing 2,000 euros (``Nichoirs pour mésanges'') were proposed (see Pabulib~\citep{pabulib}).
	It seems unlikely that both projects offer the same utility to their supporters.
	But full proportionality of utility and cost also seems implausible, because the cost effectiveness
	of different projects can vary widely. Consider, for example, a scenario where two parks of
	equal size could be built in different neighbourhoods. Now, it might be more expensive to
	build the park in one neighbourhood due to higher property prices. In that case,
	there is no reason to assume that the more expensive park offers more utility to its
	supporters.
	Crucially, these two examples show that higher cost sometimes implies higher utility,
	while sometimes it does not. This makes it hard to imagine a way of
	estimating utilities in a principled way that works for both examples.
	
	To circumvent these difficulties 
	we propose to develop
	fairness measures that are not based on \emph{equality of welfare}
	but that instead aim for \emph{equality of resources}, an idea first proposed by Ronald Dworkin~\citep{Dworkin1,Dworkin2}. In other words, we do not aim for a fair distribution of
	satisfaction, but instead we strive to invest the same effort into satisfying each voter.
	The advantage is that the amount of resources spent is a quantity we can measure objectively.
	This idea can be formalised through the notion of \emph{share} \citep{LMR21}.
	Ideally, we want to find a budget allocation where each voter has the same share.
	Let us emphasise that we do not interpret the share as a measure of satisfaction, but rather of a distribution of resources.
	Interestingly, fairness notions based on share also provide an explanation on
	how each voter's part of the budget was spent. In contrast to the related notion
	of \emph{priceability}~\citep{PS20}, here all supporters of a project ``contribute'' the same amount. As priceability allows for unequal contribution of voters, it does not
	qualify as a notion based on equality of resources.
	
	In this paper, we investigate the viability of the share as a basis of
	fairness notions in PB in several complementary ways.
	First, we propose several axioms that formalise what it means for an outcome
	to be fair in terms of share. We observe that it is not always possible to guarantee
	everyone their \emph{fair share}, which we define as
	the budget divided by the numbers of voters. For this reason, we consider several relaxations,
	such as the \emph{justified share}, where we only aim to allocate to a voter the resources they
	deserve by virtue of being part of a cohesive group.
	Moreover, we identify a version of MES~\citep{PPS21NeurIPS},
	that satisfies all share-based axioms known to be satisfiable by a tractable voting rule.
	Finally, using data from a large number of real-life PB exercises~\citep{pabulib}, we analyse to what extent it is possible to provide voters with their fair share in practice
	and how well established PB voting rules meet share-based desiderata.
	
	\paragraph{Roadmap.} After introducing our model in Section~\ref{sec:preliminaries}, we investigate the fair share in Section~\ref{sec:fair-share} and the justified share in Section~\ref{sec:justifiedShare}.
	Relationships between the concepts are discussed in Section~\ref{sec:taxonomy},
	while Section~\ref{sec:experiments} reports on an experimental study. 
	Some proofs only appear in the appendix.

	\section{The Model}
	\label{sec:preliminaries}
	
	A PB problem is described by an \emph{instance} $I = \tuple{\projSet, c, b}$ where $\projSet$ is the set of available \emph{projects}, $c: \projSet \rightarrow \mathbb{N}$ is the \emph{cost function}---mapping any project $p \in \projSet$ to its cost $c(p) \in \mathbb{N}$---and $b \in \mathbb{N}$ is the \emph{budget limit}. We write $c(P)$ instead of $\sum_{p \in P} c(p)$ for sets of projects $P \subseteq \projSet$.
	If $c(p)=1$ for all $p \in \projSet$, then $I$ belongs to the \emph{unit-cost setting}.
	
	Let $\agentSet = \{1, \ldots, n\}$ be a set of \emph{agents}. 
	When facing a PB instance, each agent is asked to submit a (not necessarily feasible) \emph{approval ballot} representing the subset of projects they approve of. The approval ballot of agent $i \in \agentSet$ is denoted by $A_i \subseteq \projSet$, and the resulting vector $\profile = (A_1, \ldots, A_n)$ 
	is called a \emph{profile}. We assume without loss of generality that every project is approved by at least one agent.
	
	Given an instance $I = \tuple{\projSet, c, b}$, we need to select a set of projects $\pi \subseteq \projSet$ to implement. Such a \emph{budget allocation}~$\pi$ has to be \emph{feasible}, \textit{i.e.}, we require $c(\pi) \leq b$. Let $\allocSet(I) = \{\pi \subseteq \projSet \mid c(\pi) \leq b\}$ be the set of feasible budget allocations for~$I$.
	
	Choosing allocations is done by means of (resolute) \emph{PB rules}. Such a rule $\pbRule$ is a function that maps any instance $I$ and profile $\profile$ over $I$ to a single feasible budget allocation $\pbRule(I, \profile) \in \allocSet(I)$. Whenever ties occur between several outcomes, we assume that they are broken in a fixed and consistent manner (\textit{e.g.}, lexicographically).
	
	We are going to propose several fairness properties we might want a rule to satisfy. All of these properties will be defined in terms of the fundamental notion of an agent's \emph{share}.
	
	\begin{definition}[Share]
		Given an instance $I = \tuple{\projSet, c, b}$ and a profile~$\profile$, the share of agent~$i$ for a subset of projects $P \subseteq \projSet$ is defined as follows:
		\[\share(I, \profile, P, i) = \sum_{p \in P \cap A_i} \frac{c(p)}{|\{A \in \profile \mid p \in A\}|}\]
	\end{definition}
	
	\noindent
	When clear from context, we shall omit the arguments of $I$ and~$\profile$.
	We interpret an agent's share as the amount of resources spent by the decision maker on satisfying the needs of that agent.
	It is important to note that the share cannot be captured via independent cardinal utility functions as the share of an agent depends on the ballots submitted by the other agents.
	Let us illustrate the concept of share on an example.
	
	\begin{example}
		Consider a PB instance with three projects such that $c(p_1) = 8$ and $c(p_2) = c(p_3) = 2$, and a budget limit $b = 8$. The profile $\profile$ is composed of four ballots such that $A_1 = A_2 = \{p_1, p_2\}$, $A_3 = \{p_1\}$ and $A_4 = \{p_3\}$.
		The most commonly used PB rule, which greedily picks the most approved projects,
		would select the bundle $\{p_1,p_2\}$. This gives agents $1$ and $2$ a share of
		$\nicefrac{6}{3} + \nicefrac{2}{2} = 3$, agent $3$ a share of $\nicefrac{6}{3} = 2$
		and agent $4$ a share of $0$. We claim that $\{p_1,p_3\}$ would be a fairer bundle
		as it gives agent $1$, $2$ and $3$ a share of $\nicefrac{6}{3} = 2$ and agent $4$ a 
		share of $\nicefrac{2}{1}= 2$. Hence, all agents have the same share.
	\end{example}
	
	\noindent
	In the sequel, we will introduce several properties of budget allocations. We shall extend them to properties of rules so that rule $F$ is said to satisfy fairness property $\mathcal{F}$ defined for budget allocations if, for every $I$ and $\profile$, the outcome $F(I,\profile)$ satisfies~$\mathcal{F}$.
	
	\section{Fair Share}
	\label{sec:fair-share}
	
	The first fairness property we study is based on the idea that
	each voter deserves $\nicefrac{1}{n}$ of the budget---a fundamental idea familiar, for instance, from the classical fair division (``cake cutting'') literature~\citep{RobertsonWebb1998}.
	So a perfect allocation would give each voter a share of $\nicefrac{b}{n}$ (unless they do not approve of enough projects for this to be possible).
	
	\begin{definition}[Fair Share]
		Given an instance $I = \tuple{\projSet, c, b}$ and a profile $\profile$, the fair share of agent $i \in \agentSet$ is defined as:
		\[\fairshare(i) = \min\{ \nicefrac{b}{n}, \share(A_i, i) \}.\]
		A budget allocation $\pi \in \allocSet(I)$ is said to satisfy fair share (FS) if for every agent we have $\share(\pi, i) \geq \fairshare(i)$.
	\end{definition}
	
	\noindent It is easy to see that for some instances, no budget allocation would provide a fair share, and thus no rule can possibly satisfy FS. Take for instance two projects of cost~1, a budget limit of~1 and two agents each approving of a different project. Then, both agents have a fair share of $\min\{\nicefrac{1}{2}, 1\} = \nicefrac{1}{2}$. However, whichever project is selected (at most one can be selected), the share of one agent would be~0.
	
	Even more, we can show that no polynomial-time computable rule can return an FS allocation whenever one exists. Indeed, checking whether an FS allocation exists is \complexNP-complete.
	
	\begin{proposition}\label{prop:FS-NP}
		Given an instance $I = \tuple{\projSet, c, b}$ and a profile~$\profile$, checking whether there exists a feasible budget allocation $\pi \in \allocSet(I)$ that satisfies FS is \complexNP-complete, even in the unit-cost setting.
	\end{proposition}
	
	\noindent
	The proof involves a reduction from \textsc{3-Set-Cover} \citep{FurerY11}.
	Due to these shortcomings of FS, we introduce two relaxations that are inspired by
	relaxations of important, satisfaction based fairness axioms,
	Extended Justified Representation up to one project (EJR-1)~\citep{PPS21NeurIPS}
	and Local-Budget Proportional Justified Representation (Local BPJR)~\citep{ALT18}.
	
	\begin{definition}[FS up to one project]
		Given an instance $I = \tuple{\projSet, c, b}$ and a profile $\profile$, a budget allocation $\pi \in \allocSet(I)$ is said to satisfy fair share up to one project (FS-1) if, for every agent~$i$, there is a project $p \in \projSet$ such that:
		\[\share(\pi \cup \{p\}, i) \geq \fairshare(i).\]
	\end{definition}
	
	\noindent
	Thus, FS-1 requires that every agent is only one project away from their fair share.
	Unfortunately, FS-1 is not always satisfiable. 
	
	\begin{proposition}
		There exist instances $I$ for which no budget allocation $\pi \in \allocSet(I)$ provides FS-1.\label{prop:FS-1-may-not-exist}
	\end{proposition}
	
	\begin{proof}
		\ifshort
		Consider an instance with three projects of cost 3 and a budget limit $b = 5$. Consider three agents, with approval ballots $\{p_1, p_2\}$, $\{p_1, p_3\}$ and $\{p_2, p_3\}$, respectively.
		\fi

		Here the fair share of each agent is $\nicefrac{5}{3} \approx 1.67$. As a single project only yields a share of $1.5$ to an agent who approves of it, for any agent to reach their fair share threshold, two projects must be selected. However, a feasible budget allocation can select at most one project, meaning that for one agent none of the projects they approve of will be selected.
	\end{proof}
	
	\noindent As for FS, we can show that deciding whether an FS-1 budget allocation exists is \complexNP-complete.
	
	\begin{proposition}
		\label{prop:FS1_NP_complete}
		Given an instance $I = \tuple{\projSet, c, b}$ and a profile $\profile$, checking whether there exists a feasible budget allocation $\pi \in \allocSet(I)$ that satisfies FS-1 is \complexNP-complete, even in the unit-cost setting.
	\end{proposition}
	
	\noindent 
	Alternatively, we can require that every project that is not part of the winning
	budget allocation should give some voter at least their fair share when that project is added.%
	\footnote{We stress that this formulation of Local-FS
		relies on our assumption that every project~$p$ is approved by at least one agent.}
	
	\begin{definition}[Local-FS]
		Given an instance $I = \tuple{\projSet, c, b}$ and a profile $\profile$, a budget allocation $\pi \in \allocSet(I)$ is said to satisfy local fair share (Local-FS) if there is no project $p \in \projSet \setminus \pi$ such that, for all agents $i \in \agentSet$ with $p \in A_i$, we have:
		\[\share(\pi \cup \{p\}, i) < \fairshare(i).\]
	\end{definition}
	
	\noindent
	Intuitively, if there exists a project $p$ that could be added to the budget allocation $\pi$ without any supporter of $p$ receiving at least their fair share, then
	every supporter of $p$ receives strictly less than their fair share and
	one of the following holds:
	
	\begin{itemize}[noitemsep]
		\item $p$ can be selected without exceeding the budget limit $b$;
		\item some voter $i^\star$ receives more than their fair share.
	\end{itemize}
	
	\noindent In the first case, it is clear that $p$ should be selected and thus $\pi$ must be deemed unfair.
	In the second case,
	it might be considered fairer to exchange one
	project supported by $i^\star$ with $p$.
	In this sense, the property can be seen as an ``upper quota'' property, as we
	have to add projects such that no voter receives more than their fair share as long as possible.
	
	In contrast to FS-1, we can always find an allocation that satisfies Local-FS.
	Indeed, an adaption of the \emph{Method of Equal Shares} (MES)\footnote{The rule used to be named Rule~X until recently. The new name---\emph{method of equal
			share}---is not related to the definition of share, first introduced by \citet{LMR21}.}
	\citep{PPS21NeurIPS}
	satisfies Local-FS.
	Our definition closely resembles the definition
	of MES for PB with additive utilities~\citep{PPS21NeurIPS}.
	We adapt it by plugging in the share.
	Note that this rule can be executed in polynomial time.
	
	\begin{definition}[\rulex{}]
		\label{def:mes_share}
		Given an instance $I=\tuple{\projSet, c, b}$ and a profile $\profile$, \rulex{} constructs a budget allocation $\pi$, initially empty,
		iteratively as follows.
		A load $\ell_i: 2^\projSet \rightarrow \Rplus$, is associated with every agent $i \in \agentSet$, initialised as
		$\ell_i(\emptyset) = 0$ for all $i \in \agentSet$.
		It represents the total contribution of the agents for a given budget allocation.
		Given $\pi$ and a scalar $\alpha \geq 0$, the contribution of agent $i \in \agentSet$ for project
		$p \in \projSet \setminus \pi$ is defined by:
		\[\gamma_i(\pi, \alpha, p) = \min\left(\nicefrac{b}{n} - \ell_i(\pi), \alpha \cdot \share(\{p\}, i)\right).\]
		For a specific budget allocation $\pi$, a project $p \in \projSet \setminus \pi$ is said to be $\alpha$-affordable,
		for $\alpha \geq 0$, if $\sum_{i \in \agentSet} \gamma_i(\pi, \alpha, p) \cdot \mathds{1}_{p \in A_i} = c(p)$, \textit{i.e.}, if $p$ can be afforded with none of its supporters contributing more than $\alpha$.
		
		At a given round with current budget allocation $\pi$, if no project is $\alpha$-affordable for any $\alpha$,
		\rulex{} terminates.
		Otherwise, it selects a project $p \in \projSet \setminus \pi$ that is $\alpha^\star$-affordable where $\alpha^\star$
		is the smallest $\alpha$ such that one project is $\alpha$-affordable ($\pi$ is updated to $\pi \cup \{p\}$).
		The agents' loads are then updated: If $p \notin A_i$, then $\ell_i(\pi \cup \{p\}) = \ell_i(\pi)$, and
		otherwise $\ell_i(\pi \cup \{p\}) = \ell_i(\pi) + \gamma_i(\pi, \alpha, p)$.
	\end{definition}
	
	\begin{theorem}
		\label{thm:Rule-X-Local-FS}
		\rulex{} satisfies Local-FS.
	\end{theorem}
	
	\begin{proof}
		Given a budget allocation $\pi$ and a scalar $\alpha > 0$, we say that agent $i \in \agentSet$ contributes in full to project $p \in A_i$ if we have: $\gamma_i(\pi, \alpha, p) = \alpha \cdot \share(\{p\}, i)$.
		
		During a run of \rulex{}, all the supporters of a project $p \in \projSet$ contribute in full to $p$ if and only if $p$ is $1$-affordable.
		In this case, for all supporters $i$ of $p$, we have $\ell_i(\{p\}) = \share(\{p\}, i)$.
		Moreover, if a project $p$ is $\alpha$-affordable but at least one voter cannot contribute in full to $p$,
		then $\alpha > 1$. \rulex{} only terminates when no project is $\alpha$-affordable, for any $\alpha$.
		Therefore, there is a round where no project $p$ is $1$-affordable. Let $k$ be
		the first such round and let $\pi_k$ be the budget allocation before round $k$.
		It follows that every project in $\pi_k$ was $1$-affordable and hence
		$\ell_i(\pi_k) = \share(\pi_k, i)$ for all $i \in \agentSet$.
		As no project $p$ is $1$-affordable in round $k$, for no projects in $\projSet \setminus \pi_k$ can all the supporters contribute in full to. Thus, for every $p\in \projSet \setminus \pi_k$,
		there is a voter $i \in \agentSet$ such that
		$\nicefrac{b}{n} - \ell_i(\pi_k) < \share(\{p\}, i)$.
		Using the fact that $\ell_i(\pi_k) = \share(\pi_k, i)$ and the additivity of share, it follows that
		$\share(\pi_k \cup \{p\}, i) > \nicefrac{b}{n}$.
		So $\pi_k$ satisfies Local-FS. As \rulex{} returns an allocation $\pi$ with $\pi_k \subseteq \pi$, it satisfies Local-FS.
	\end{proof}
	
	\noindent
	In fact, the proof of Theorem~\ref{thm:Rule-X-Local-FS} establishes a slightly stronger statement:
	there is no project $p \in \projSet \setminus \pi$ such that for all agents $i \in \agentSet$ with $p \in A_i$ we have $\share(\pi \cup \{p\}, i) \leq \nicefrac{b}{n}$.
	In other words, any project added to $\pi$ gives at least one voter \emph{more} than
	their fair share.

	\section{Justified Share}
	\label{sec:justifiedShare}
	
	Local-FS and FS-1 require the outcome to be, in some sense, close to satisfying FS. 
	Another idea for weakening FS is to spend on a voter
	only the resources 
	they 
	can claim to deserve by virtue of being part of a cohesive group. 
	This idea is inspired by the well-known axioms of justified representation
	extensively studied in the literature on approval-based committee elections \citep{ABCEFW17,AEHLSS18,PS20,LaSk23}.
	Before exploring this idea further, let us define what we mean by cohesive groups.
	
	\begin{definition}[$P$-cohesive groups]
		Given an instance $I = \tuple{\projSet, c, b}$ and a profile $\profile$, for a set of projects $P \subseteq \projSet$ we say that a non-empty group of agents $N \subseteq \agentSet$
		is $P$-cohesive, if $P \subseteq \bigcap_{i \in N} A_i$ and $\frac{|N|}{n} \geq \frac{c(P)}{b}$.
	\end{definition}
	
	\noindent
	So a group~$N$ is cohesive relative to a set~$P$ of projects if, first, everyone in $N$ approves of all the projects in $P$ and, second, $N$ is large enough---relative to the size~$n$ of the society and the budget~$b$---so as to ``deserve'' the resources needed for the projects in~$P$.
	
	In the unit-cost setting, one of the strongest proportionality properties known to 
	be satisfiable by a polynomial-time-computable rule is called Extended Justified Representation (EJR) \citep{AEHLSS18,PS20}. 
	It ensures that one member of every cohesive group receives the satisfaction deserved due to their cohesiveness.
	\citet{PPS21NeurIPS} generalised EJR to the 
	setting of PB with additive utilities. 
	This generalisation will be our blue-print for modifying EJR
	to deal with share. 
	Ideally, we would want to satisfy the following property.
	
	\begin{definition}[Strong Extended Justified Share]
		Given an instance $I = \tuple{\projSet, c, b}$ and a profile $\profile$, a budget allocation $\pi \in \allocSet(I)$ is said to satisfy strong extended justified share (Strong-EJS) if for all $P \subseteq \projSet$ and all $P$-cohesive groups $N$, we have $\share(\pi, i) \geq \share(P, i)$ for all agents $i \in N$.
	\end{definition}
	
	\noindent
	The idea behind Strong-EJS is the following: since every $P$-cohesive group $S$ controls enough budget to fund $P$, every agent in $S$ deserves to enjoy at least as much share as what they would have gotten if $P$ had been the outcome.
	Intuitively, this is very similar to Strong-EJR, a property that is known not to be always 
	satisfiable~\citep{ABCEFW17}. The same holds for Strong-EJS: there exist instances for which no budget allocation can satisfy this axiom.

	\begin{example}\label{Exp:PJS-impos}
		Consider the following instance and profile with three projects $p_1, p_2$ and $p_3$ of cost $1$, a budget limit $b = 2$, and four agents $1, \ldots, 4$ such that $1$ approves project $p_1$, $2$ approves project $p_1$ and $p_2$, $3$ approves $p_1$ and $p_3$ and $4$ approves $p_2$ and $p_3$.
		
		Note that $\{1, 2, 3\}$ is $\{p_1\}$-cohesive, $\{2, 4\}$ is $\{p_2\}$-cohesive and $\{3, 4\}$ is $\{p_3\}$-cohesive. Hence, to satisfy Strong-EJS, one needs to select all three projects which is not possible within the given budget limit.
	\end{example}
	
	\noindent
	Observe that in this scenario it is not even possible to guarantee each $P$-cohesive
	group the same \emph{average} share as they receive from~$P$.
	We thus weaken Strong-EJS and introduce (simple) EJS.
	
	\begin{definition}[Extended Justified Share]
		Given an instance $I = \tuple{\projSet, c, b}$ and a profile $\profile$, a budget allocation $\pi \in \allocSet(I)$ is said to satisfy extended justified share (EJS), if for all $P \subseteq \projSet$ and all $P$-cohesive groups $N$, there is an agent $i \in N$ such that $\share(\pi, i) \geq \share(P, i)$.
	\end{definition}
	
	\noindent The difference between Strong-EJS and EJS is the switch from a universal to an existential quantifier: for the former, we impose a lower bound on the share of \emph{every} agent in a cohesive group, while for the latter this lower bound only applies to one agent of each cohesive group.
	Therefore, in Example~\ref{Exp:PJS-impos} both $\{p_1,p_3\}$ and $\{p_2,p_3\}$ satisfy 
	EJS, as either agent~$3$ or agent~$4$ satisfies the share requirement for every cohesive group.
	
	We observe that EJR and EJS, while similar in spirit, do not coincide, not even in the unit-cost case.

	\begin{example}\label{EJS-EJR}
		Consider an instance with four voters and six projects with unit cost and $b = 4$,
		where the approvals are as follows:
		$A_1 =\{p_1,p_2,p_3\}$, $A_2 =\{p_1,p_2,p_4\}$, $A_3 = A_4 = \{p_4,p_5,p_6\}$.
		
		It is now easy to check that $\{p_3,p_4,p_5,p_6\}$ satisfies EJS but not EJR, while $\{p_1,p_4,p_5,p_6\}$ satisfies EJR but not EJS.
	\end{example}
	
	\noindent
	The first question that presents itself 
	is whether EJS is always achievable. This is indeed the case. To see this, one just needs
	to adapt the well-known greedy cohesive procedure for satisfying EJR, which was first introduced
	by \citet{ABCEFW17} and extended to PB by \citet{PPS21NeurIPS}, to the share setting. 
	
	\begin{proposition}\label{Prop:Sat-EJS}
		For every instance $I = \tuple{\projSet, c, b}$ and every profile $\profile$, there exists a budget allocation $\pi \in \allocSet(I)$ that satisfies EJS.
	\end{proposition}
	
	\noindent However, the greedy approach in general needs exponential time.
	This turns out to be unavoidable, unless $\complexP = \complexNP$,
	as can be shown by a standard reduction from \textsc{Subset Sum}.
	
	\begin{theorem}\label{Thm:EJS-hard}
		There is no polynomial-time algorithm that, given an instance $I$ and a profile $\profile$ as input, always computes a budget allocation satisfying EJS, unless $\complexP = \complexNP$.
	\end{theorem}
	
	\noindent
	On the other hand, we recall that the greedy approach generally runs in
	\complexFPT{}-time, when parameterized by the number of projects~\citep{ABCEFW17}.
	This is also the case in the share setting.
	
	\begin{proposition}
		\label{prop:runningTimeGreedyEJS}
		For every instance $I = \tuple{\projSet, c, b}$ and every profile $\profile$,
		we can compute a budget allocation $\pi \in \allocSet(I)$ that satisfied EJS
		in time $\bigO(n \cdot 2^{|\projSet|})$.
	\end{proposition}

	\ifshort\noindent\fi We have seen that EJS can always be satisfied. However, this is not entirely satisfactory, given that
	no tractable rule can satisfy it. Unfortunately, in many PB applications, the
	use of intractable rules is not practical due to the large instance sizes.
	Therefore, we try to find fairness notions that can be satisfied in polynomial time
	by relaxing EJS.
	
	First, similar to the property of EJR up to one project (EJR-1) proposed by \citet{PPS21NeurIPS},
	we can define EJS up to one project, requiring that at least one agent in every cohesive
	group is at most one project away from being satisfied.\footnote{We note that in Definition~\ref{def:ejs-1} we require that $\share(\pi \cup \{p\}, i) \geq \share(P, i)$ instead of a strict inequality as used in the definition of EJR-1. Our rationale is that adding one project guarantees to satisfy the EJS condition (but not more than that).}
	
	\begin{definition}[EJS-1]\label{def:ejs-1}
		Given an instance $I = \tuple{\projSet, c, b}$ and a profile $\profile$, a budget allocation $\pi \in \allocSet(I)$ is said to satisfy extended justified share up to one project (EJS-1) if for all $P \subseteq \projSet$ and all $P$-cohesive groups $N$ there is an agent $i \in N$ for which there exists a project $p \in \projSet$ such that $\share(\pi \cup \{p\}, i) \geq \share(P, i)$.
	\end{definition}
	
	\noindent It is straightforward to adapt the proof of \citet{PPS21NeurIPS}
	that MES satisfies EJR up to one project to our setting to prove that
	\rulex{} satisfies EJS up to one project.
	
	\begin{proposition}\label{prop:Rule-X-EJS-1}
		\rulex{} satisfies EJS-1.
	\end{proposition}
	
	\noindent In particular, this implies, together with Example~\ref{EJS-EJR}, that already in the unit-cost case MES and \rulex{} are indeed
	different rules.
	
	Finally, note that we can define a local variant of EJS, based on a similar motivation as Local-FS.
	
	\begin{definition}[Local-EJS]
		Given an instance $I = \tuple{\projSet, c, b}$ and a profile $\profile$, a budget allocation $\pi \in \allocSet(I)$ is said to satisfy local extended justified share (Local-EJS), if there is no $P$-cohesive group $N$, where $P \subseteq \projSet$, for which there exists a project $p \in P \setminus \pi$ for which it holds for all agents $i \in N$ that $\share(\pi \cup \{p\}, i) < \share(P, i)$.
	\end{definition}
	
	\noindent The idea behind Local-EJS is that there is no $P$-cohesive group $N$ that can claim that they could ``afford''
	another project $p$ without a single voter in $N$ receiving more share than they deserve due to their $P$-cohesiveness.
	In this sense, any allocation that satisfies Local-EJS is a local optimum for any $P$-cohesive group.
	Now, in our setting we observe that Local-EJS is equivalent to a notion that could be called ``EJS up to any project''.
	
	\begin{proposition}
		\label{prop:LocalEJS<=>EJSX}
		Let $I = \tuple{\projSet, c, b}$ be an instance and $\profile$ a profile.
		An allocation $\pi$ satisfies Local-EJS if and only if for every $P \subseteq \projSet$ and
		$P$-cohesive group $N$ there exists an agent~$i$ such that for all projects
		$p \in P \setminus \pi$ we have $\share(\pi \cup \{p\}, i) \geq \share(P, i)$.
	\end{proposition}
	
	\begin{proof}
		It is clear that the statement above implies Local-EJS. Now, let $\pi$ be an allocation
		that satisfies Local-EJS, let $P \subseteq \projSet$ be a set of projects and $N$
		a $P$-cohesive group. Let $i^* \in N$ be an agent with maximal share from $\pi$ in $N$.
		Consider $p \in P \setminus \pi$. By Local-EJS there is an agent $i_p$ such that
		$\share(\pi\cup\{p\}, i_p) > \share(P, i_p)$. By the choice of $i^*$ we have
		$\share(\pi, i^*) \geq \share(\pi, i_p)$. By the definition of share, it follows that
		$\share(\pi\cup\{p\}, i^*) > \share(P, i^*)$.
	\end{proof}
	
	\noindent
	From this equivalence, it is easy to see that Local-EJS implies EJS-1.
	Unfortunately, \rulex{} fails Local-EJS, as the next example shows.

	\begin{example}
		Consider an instance with five projects, a budget limit $b = 20$, and four agents where the costs are as follows:
		\[c(p_1) = 8, c(p_2) = 5, c(p_3) = c(p_4) = 2, c(p_5) = 10.\]
		Moreover, voters $1$ and $2$ approve projects $p_1$, $p_2$, $p_3$ and $p_4$ and
		voters $3$ and $4$ prove $p_3,p_4$ and $p_5$.
		With a suitable tie-breaking rule, \rulex{} will return the budget allocation $\pi = \{p_2, p_3, p_5\}$. Note that voters 1 and 2 are $\{p_1, p_4\}$-cohesive and would thus deserve to enjoy a share of 4.5. However, if we add $p_4$ to $\pi$, voters 1 and 2 would only have a share of 3.5, showing that $\pi$ fails Local-EJS.
	\end{example}
	
	\noindent
	Whether Local-EJS can always be satisfied in polynomial time remains an
	important open question.
	
	Finally, we observe a crucial difference between EJR and EJS:
	\rulex{} does not satisfy EJS in the unit cost setting, while MES satisfies EJR in the unit-cost setting \citep{PS20}.
	
	\begin{example}
		Assume that there are two voters 1 and 2, and three projects $p_1$, $p_2$ and $p_3$, all of cost 1. The budget limit is $b = 2$. Voter~1 approves of $p_1$ and $p_3$ and voter 2 of $p_2$ and $p_3$. Then voter 1 is $\{p_1\}$-cohesive and hence deserves a share of 1, the same applies to voter 2 and $\{p_2\}$. Nevertheless, with a suitable tie-breaking rule, \rulex{} would first select $p_3$. In that case. neither $\{p_1, p_2\}$, nor $\{p_2, p_3\}$ would satisfy EJS, as at least one voter will have only a share of $\nicefrac{1}{2}$.
	\end{example}
	
	\noindent
	However, it does satisfy Local-EJS in the unit-cost setting.
	
	\begin{theorem}\label{Thm:Rule-X-Local-EJS}
		\rulex{} satisfies Local-EJS in the unit-cost case.
	\end{theorem}

	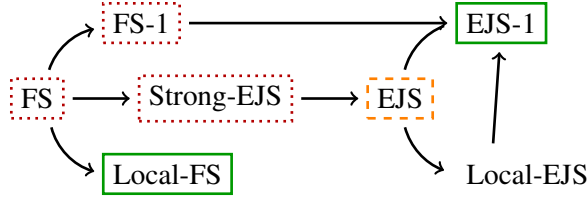
\begin{figure}
		\centering
		\begin{tikzpicture}[shorten < = 3pt, shorten > = 3pt, line width = 1pt]
			\node[draw = MyRed, dotted] (FS) at (0, -1) {FS};
			
			\node[draw = MyRed, dotted, right = 2.5em of FS] (Strong-EJS) {Strong-EJS};
			\node[draw = MyOrange, dashed, right = 2.5em of Strong-EJS] (EJS) {EJS};
			
			\node[anchor = west] (Local-EJS) at (5.5, -2) {Local-EJS};
			\node[draw = MyGreen, anchor = west] (EJS-1) at (5.5, 0) {EJS-1};
			
			\path[->] (Strong-EJS) edge (EJS);
			\path[-] (EJS) edge[bend left = 27] (5.45, 0.001);
			\path[->] (EJS) edge[bend right = 27] (5.45, -1.99);
			\path[<-] (EJS-1) edge (5.99, -1.7);
			
			\node[draw = MyGreen, anchor = west] (Local-FS) at (0.85, -2) {Local-FS};
			\node[draw = MyRed, dotted, anchor = west] (FS-1) at (0.85, 0) {FS-1};
			
			\path[->] (FS) edge (Strong-EJS);
			\path[->] (FS) edge[bend left = 30] (FS-1);
			\path[->] (FS) edge[bend right = 30] (Local-FS);
			\path[->] (FS-1) edge (EJS-1);
		\end{tikzpicture}
		\caption{Taxonomy of criteria. An arrow from one criterion to another indicates that any budget allocation satisfying the former also satisfies the latter. \rulex{} satisfies the criteria boxed in green solid lines. For the criterion boxed in orange dashed lines, no efficient algorithms computing them exist (unless $\complexP = \complexNP$). Criteria boxed in red dotted lines are not always satisfiable. The status of Local-EJS is unknown.}
		\label{fig:taxonomy}
	\end{figure}

	\section{Relationships between Criteria}
	\label{sec:taxonomy}
	
	We now analyse the relationships between different fairness criteria. We start with links within the space of criteria we introduced earlier and then compare them to the notion of priceability~\citep{PS20}.
	
	\subsection{Share-Based Fairness Criteria}\label{sec:taxonomy-share}
	
	The following theorem establishes the relations between share-based fairness concepts. These relations are visualised in Figure~\ref{fig:taxonomy}.
	
	\begin{theorem}
		\label{thm:share-all-relations}
		Given an instance $I$ and a profile $\profile$, for every budget allocation $\pi \in \allocSet(I)$ the following statements hold:
		\begin{enumerate}
			\item If $\pi$ satisfies FS, it also satisfies FS-1, Local-FS, and Strong-EJS.\label{relations-fs-implies}
			\item If $\pi$ satisfies FS-1, it also satisfies EJS-1.\label{relations-fs1-implies-ejs1}
			\item If $\pi$ satisfies Strong-EJS, it also satisfies EJS.\label{relations-strongejs-implies-es1}
			\item If $\pi$ satisfies EJS, it also satisfies Local-EJS.	\label{relations-ejs-implies-localejs}
			\item If $\pi$ satisfies Local-EJS, it also satisfies EJS-1.\label{relations-localejs-implies-ejs1}
		\end{enumerate}
		This list of implications is exhaustive when closed under transitivity.
	\end{theorem}
	
	\begin{proof}[Proof of \eqref{relations-fs-implies}.]
		
		It is easy to verify that every budget allocation satisfying FS also satisfy FS-1 and Local-FS. So let us show that FS also implies Strong-EJS.
		Let $i \in \agentSet$. We distinguish two cases. 
		
		First, assume $\share(A_i, i) < \nicefrac{b}{n}$. For FS to be satisfied, we must have $\share(\pi, i) \geq \share(A_i, i)$. This entails that $A_i \subseteq \pi$. Hence, the conditions for Strong-EJS are trivially satisfied for agent~$i$.
		
		Second, assume $\share(A_i, i) \geq \nicefrac{b}{n}$. Since $\pi$ satisfies FS, we know that $\share(\pi, i) \geq \nicefrac{b}{n}$. Let $N \subseteq \agentSet$ be a $P$-cohesive group, for some $P \subseteq \projSet$, such that $i \in N$. By definition of a cohesive group, we know that $c(P) \leq \nicefrac{b}{n} \cdot |N|$. Hence, $\share(P, i) \leq \nicefrac{b}{n}$. Overall, we have
		$\share(\pi, i) \geq \nicefrac{b}{n} \geq \share(P, i)$ and thus
		$\pi$ satisfies Strong-EJS.
	\end{proof}
	
	\noindent
	The proof of claim~\eqref{relations-fs1-implies-ejs1} is similar to the proof of claim~\eqref{relations-fs-implies} and can be found in the appendix. The proofs of claims~\eqref{relations-strongejs-implies-es1} and~\eqref{relations-ejs-implies-localejs} are immediately derived from the respective definitions. The proof of claim~\eqref{relations-localejs-implies-ejs1} is a direct consequences of Proposition~\ref{prop:LocalEJS<=>EJSX}.
	
	The absence of any further implications between fairness criteria can be established by counterexamples. We include here a representative sample of such counterexamples.
	The remaining counterexamples can be found in the appendix.
	
	\begin{example}[FS-1 not implying Local-FS]
		Consider the following instance with four projects and a budget limit of $b = 6$.
		\begin{center}\small
			\begin{tabular}{ccccc}
				\toprule
				& $p_1$ & $p_2$ & $p_3$ & $p_4$ \\
				\midrule
				Cost & 3 & 3 & 6 & 1 \\
				\midrule
				$A_1$ & \cmark & \cmark & \xmark & \xmark \\
				$A_2$ & \xmark & \xmark & \cmark & \cmark \\
				$A_3$ & \xmark & \xmark & \cmark & \cmark \\
				\bottomrule
			\end{tabular}
		\end{center}
		Then $\{p_1, p_2\}$ satisfies FS-1 as agent $1$ already receives (more than)
		their
		fair share, while $2$ and $3$ receive their fair share from $\{p_1,p_2\}\cup \{p_3\}$.
		However, no supporter of $p_4$ receives their fair share from $\{p_1,p_2\}\cup \{p_4\}$.
		Therefore, Local-FS is violated.
	\end{example}
	
	\noindent
	The other direction (Local-FS does not imply FS-1) follows from the fact that a Local-FS allocation always exists (as a consequence of Theorem~\ref{thm:Rule-X-Local-FS}) while FS-1 is not always satisfiable (Proposition~\ref{prop:FS-1-may-not-exist}).
	
	\begin{example}[EJS-1 not implying Local-EJS, even in the unit-cost setting]
		Consider an instance with two voters, 1 and 2, and six projects $p_1$, \ldots, $p_6$ all of cost 1. Voter 1 approves of $\{p_1, p_2, p_3, p_4, p_5\}$ and voter~2 approves of $\{p_4, p_5 , p_6\}$. The budget limit is $b = 4$. We claim that $\pi = \{p_1, p_2, p_3, p_4\}$ satisfies EJS-1 but not Local-EJS.
		
		The share of 1 in $\pi$ is 3.5 so every cohesive group containing them will satisfy the conditions for EJS-1 and Local-EJS. Consider now voter 2. Their share in $\pi$ is $\nicefrac{1}{2}$. Note that they are $\{p_5, p_6\}$-cohesive and deserve thus a share of $\nicefrac{3}{2}$. Since $\pi \cup \{p_6\}$ would provide them a share of $\nicefrac{3}{2}$, $\pi$ satisfies EJS-1. However, $\pi \cup \{p_5\}$ would only provide 2 a share of $1$, showing that $\pi$ fails Local-EJS.
	\end{example}
	
	\noindent We now turn to Local-FS and show that it does not imply EJS-1. Due to the implications shown in Theorem~\ref{thm:share-all-relations}, Local-FS also does not imply any of Local-EJS, EJS, Strong-EJS, FS-1 and FS.
	
	\begin{example}[Local-FS not implying EJS-1]\label{ex:Local-FS does not imply EJS-1}
		Consider the following instance with three projects, a budget limit of $b = 6$, and two agents.
		Agent~1 approves of $\{p_1,p_2,p_3\}$ and agent~2 of $\{p_2,p_3\}$.
		Allocation $\pi = \{p_1\}$ satisfies Local-FS: for both $p_2$ and $p_3$, if we were to add them to $\pi$, agent~1 would have a fair share. However, it does not satisfy EJS-1: $\{2\}$ is a $\{p_2, p_3\}$-cohesive group but neither project is selected.
	\end{example}

	\subsection{Comparison with Priceability}
	
	Priceability is a fairness criterion requiring that the budget allocation can be obtained through a market-based approach~\citep{PS20}. It is similar in spirit to share-based criteria as it also measures the amount of money spent on each agent. However, priceability does not require the cost of a project to be equally distributed between its supporters. Instead it requires there to be some distribution of the costs of the selected projects to their supporters that satisfies certain conditions. 
	
	\begin{definition}[Priceability]
		Given an instance $I = \tuple{\projSet, c, b}$ and a profile $\profile$, a budget allocation $\pi$ satisfies priceability if there exists an allowance $\alpha \in \Rplus$ and a collection $(\gamma_i)_{i \in \agentSet}$ of contribution functions, $p_i: \projSet \to [0, \alpha]$ such that all of the following conditions are satisfied:
		\begin{itemize}
			\item[\textbf{C1}:] If $\gamma_i(p) > 0$ then $p \in A_i$ for all $p \in \projSet$ and $i \in \agentSet$.
			\item[\textbf{C2}:] If $\gamma_i(p) > 0$ then $p \in \pi$ for all $p \in \projSet$ and $i \in \agentSet$.
			\item[\textbf{C3}:] $\sum_{p \in \projSet} \gamma_i(p) \leq \alpha$ for all $i \in \agentSet$.
			\item[\textbf{C4}:] $\sum_{i \in \agentSet} \gamma_i(p) = c(p) $ for all $p \in \pi$.
			\item[\textbf{C5}:] $\sum_{i \in \agentSet \mid p \in A_i} \alpha^\star_i \le c(p)$ for all $p \in \projSet \setminus \pi$,
			where for any $i \in \agentSet$, $\alpha^\star_i = \alpha - \sum_{p \in \projSet} \gamma_i(p)$ is their unspent allowance.
		\end{itemize}
	\end{definition}
	
	\noindent
	Due to the similar motivation of share based fairness concepts and priceablility, it is interesting
	to understand the relationships between them. Let us first consider FS.
	
	\begin{proposition}
		There exists instances $I = \tuple{\projSet, c, b}$ and profiles $\profile$ such that there exists $\pi \in \allocSet(I)$ satisfying FS, but such that no FS budget allocation $\pi$ is priceable.
	\end{proposition}
	
	\begin{proof}
		Consider the following instance with four projects, a budget limit of $b = 9$, and three agents.
		\begin{center}\small
			\begin{tabular}{ccccc}
				\toprule
				& $p_1$ & $p_2$ & $p_3$ & $p_4$ \\
				\midrule
				Cost & 1 & 5 & 3 & 1 \\
				\midrule
				$A_1$ & \cmark & \xmark & \xmark & \xmark \\
				$A_2$ & \xmark & \cmark & \xmark & \xmark \\
				$A_3$ & \xmark & \xmark & \cmark & \cmark \\
				\bottomrule
			\end{tabular}
		\end{center}
		
		\noindent In this instance, only $\pi = \{p_1, p_2, p_3\}$ satisfies FS. For the sake of contradiction, suppose that $\pi$ is priceable with allowance $\alpha \in \Rplus$ and contribution functions $\gamma_1, \gamma_2$ and $\gamma_3$. Since only agent 2 approves of $p_2$, from conditions \textbf{C1} and \textbf{C4} we must have $\gamma_2(p_2) = 5$. Condition \textbf{C3} then implies that $\alpha \geq 5$. For similar reasons we should have $\gamma_3(p_3) = 3$ and $\gamma_3(p_1) = \gamma_3(p_2) = 0$. Condition \textbf{C2} also imposes $\gamma_3(p_4) = 0$. Overall this means that $\alpha_3^\star = \alpha - \gamma_3(p_3) \geq 2$. This is a violation of condition \textbf{C5} for agent 3 and project~$p_4$.
	\end{proof}
	
	\noindent
	Interestingly, the intuitive connection between fair share and priceablility does hold when ballots are large enough.
	
	\begin{proposition}
		For every instance $I = \tuple{\projSet, c, b}$ and profile $\profile$ such that for every agent $i \in \agentSet$, $\fairshare(i) = \nicefrac{b}{n}$, every budget allocation $\pi \in \allocSet(I)$ that satisfies FS is also priceable.\label{prop:fs->priceability}
	\end{proposition}
	
	\begin{proof}
		
		Consider a suitable instance $I = \tuple{\projSet, c, b}$ and profile $\profile$. Let $\pi \in \allocSet(I)$ be a budget allocation that satisfies FS. We claim that $\pi$ is priceable for the allowance $\alpha = \nicefrac{b}{n}$ and the contribution functions $(\gamma_i)_{i \in \agentSet}$ defined for every $i \in \agentSet$ and $p \in \projSet$ as:
		\[\gamma_i(p) =
		\begin{cases}
			\share(\{p\}, i) & \text{if } p \in A_i \cap \pi,\\
			0 & \text{otherwise}.
		\end{cases}\]
		First note that conditions \textbf{C1} and \textbf{C2} of priceability are trivially satisfied for all $i \in \agentSet$ and $p \in \projSet$.
		Now, we know that for every agent, we have $\share(\pi, i) \geq \fairshare(i) = \nicefrac{b}{n}$. Since $\sum_{i \in \agentSet} \share(\pi, i) = c(\pi)$ and $\pi$ is feasible, we must have $\share(\pi, i) = \nicefrac{b}{n}$ for all $i \in \agentSet$. Overall, we have $\sum_{p \in \projSet}\gamma_i(p) = \share(\pi, i) = \nicefrac{b}{n} \leq \alpha$, so condition \textbf{C3} also is satisfied.
		In addition, we have $\sum_{i \in \agentSet} \gamma_i(p) = \sum_{i \in \agentSet} \share(\{p\}, i) = c(p)$. Condition \textbf{C4} is thus immediately satisfied.
		Finally, as we have for every agent $\sum_{p \in \projSet} \gamma_i(p) = \share(\pi, i) = \nicefrac{b}{n} = \alpha$, 
		condition \textbf{C5} is vacuously satisfied.
	\end{proof}
	
	\noindent
	Next, we consider the relation between the weaker share-based notions and priceability. 
	The following shows that there are no other implications between priceability and
	share-based fairness, even if we assume that agents approve of enough projects.
	
	\begin{proposition}
		Local-FS, FS-1, and EJS do not imply priceability, even if $\fairshare(i) =
		\nicefrac{b}{n}$ for every agent $i \in \agentSet$.
		Vice versa, priceability does not imply Local-FS or EJS-1, even if
		$\fairshare(i) = \nicefrac{b}{n}$ for every agent $i \in \agentSet$
		and the agents have an allowance of at least~$\nicefrac{b}{n}$.
	\end{proposition}
	
	\begin{proof}
		Consider an instance with two projects with $c(p_1) = 3$ and $c(p_2) = 2$, $b=3$ and two agents such that $A_1 = \{p_1\}$ and $A_2 = \{p_2\}$
		Then $\{p_1\}$ satisfies FS-1 and Local-FS as we have 
		$\share(\{p_1,p_2\},i) > \fairshare(i)$ for both agents $i$.
		Moreover, EJS is trivially satisfied, as there are no cohesive groups.
		On the other hand, for $\{p_1\}$ to be priceable, each agent must receive an allowance of $3$.
		In this case, the fact that $p_2$ is not selected is a contradiction to \textbf{C5} as $p_2 \in A_2$ and $2$ has more than $c(p_2)$ unspent allowance.
		
		Now consider the instance with four projects such that $c(p_1) = c(p_2) = 8$ and $c(p_3) = c(p_4) = 5$ and $b = 20$. There are two agents with ballots $A_1 = \{p_1, p_2\}$ and $A_2 = \{p_2, p_3, p_4\}$.
		The bundle $\{p_1,p_2\}$ is priceable: consider the following contributions
		with an allowance of $10$ per agent: $\gamma_1(p_1) = 8$ and $\gamma_1(p) = 0$ for $p \in \{p_2, p_3, p_4\}$;
		$\gamma_2(p_2) = 8$ and $\gamma_2(p) = 0$ for $p \in \{p_1, p_3, p_4\}$. 
		However, $\{p_1,p_2\}$ does not satisfy
		Local-FS as $\share(\{p_1,p_2\}\cup\{p_3\},2) = 9 < 10 = \fairshare(2)$. Moreover,
		$2$ is $\{p_3, p_4\}$-cohesive but $\share(\{p_1,p_2\}\cup\{p\},2) = 9 < 10 = \share(\{p_3,p_4\},2)$
		for any $p \in \{p_3, p_4\}$. Hence, $\{p_1,p_2\}$ also does not satisfy EJS-1.
	\end{proof}
	
	\noindent
	However, Local-FS, EJS-1, and priceability are compatible in the sense that there 
	always exists a bundle satisfies all three, namely the output of \rulex.
	This follows directly from Theorem~\ref{thm:Rule-X-Local-FS}, Proposition~\ref{prop:Rule-X-EJS-1},
	and the fact that MES is priceable for every utility function, as was shown by 
	\citet{PPS21NeurIPS}. It remains open whether FS-1, EJS, and Local-EJS
	are compatible with priceability in this sense.
	
	\section{Approaching Fair Share in Practice}
	\label{sec:experiments}
	
	As we saw in Section~\ref{sec:fair-share}, there exist PB instances for which it is impossible to give every agent
	their fair share.
	In this section we report on an experimental study aimed at understanding how serious a problem this is.
	Our study is twofold.
	We first investigate how close to fair share we can get.
	In a second experiment, we quantify how close to this optimal value certain known PB rules get.
	
	For these experiments we use data from Pabulib~\citep{pabulib}, an online collection of real-world PB datasets.
	To be more precise, we used all instances from Pabulib with up to 65 projects, except for trivial instances, where either no project or the set of all projects are affordable.
	Three instances have been additionally omitted for the first experiment due to very high compute time.
	A total of 353 PB instances are covered by our analysis.
	
	\begin{figure*}
		\centering
		\includegraphics[width=\linewidth]{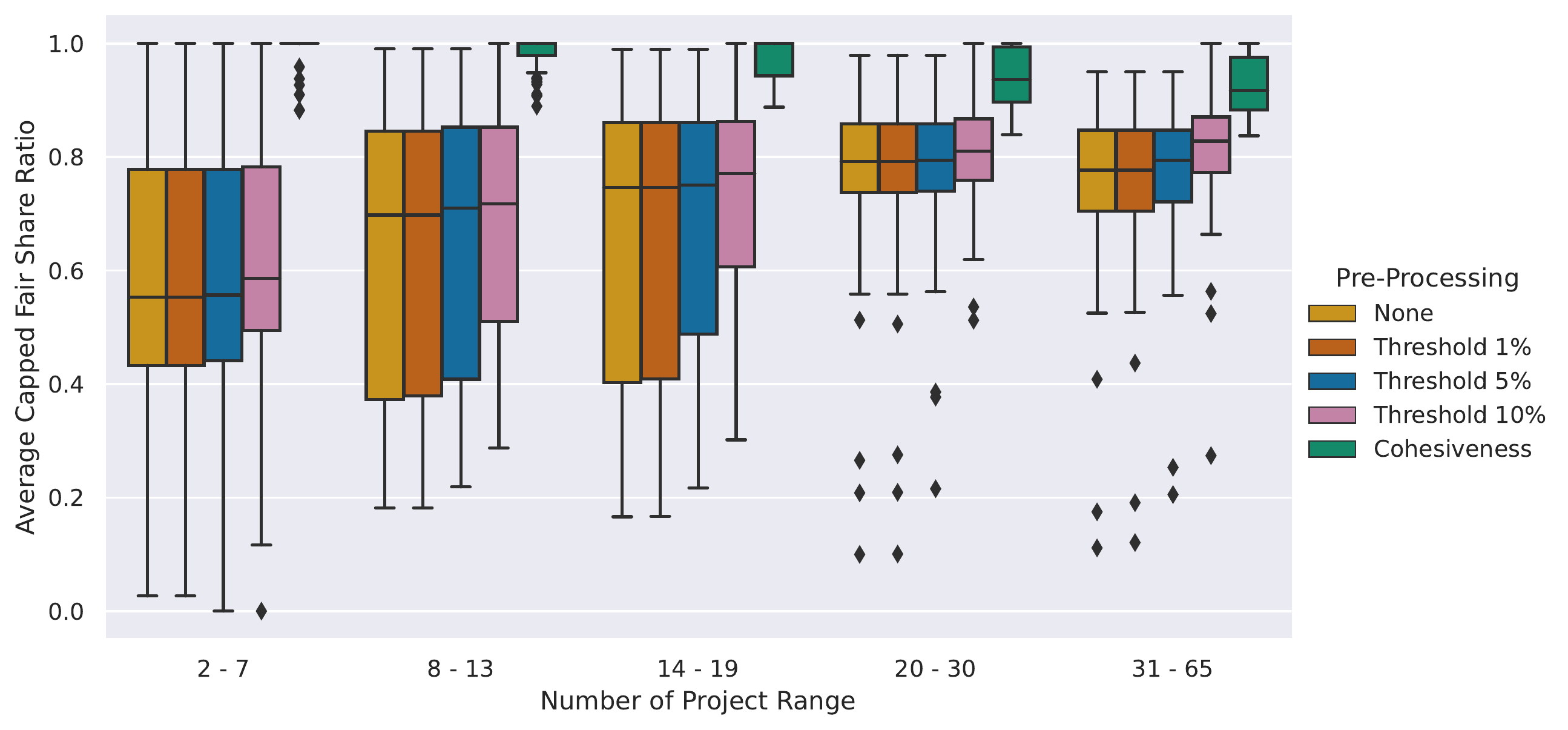}
		
		\includegraphics[width=\linewidth]{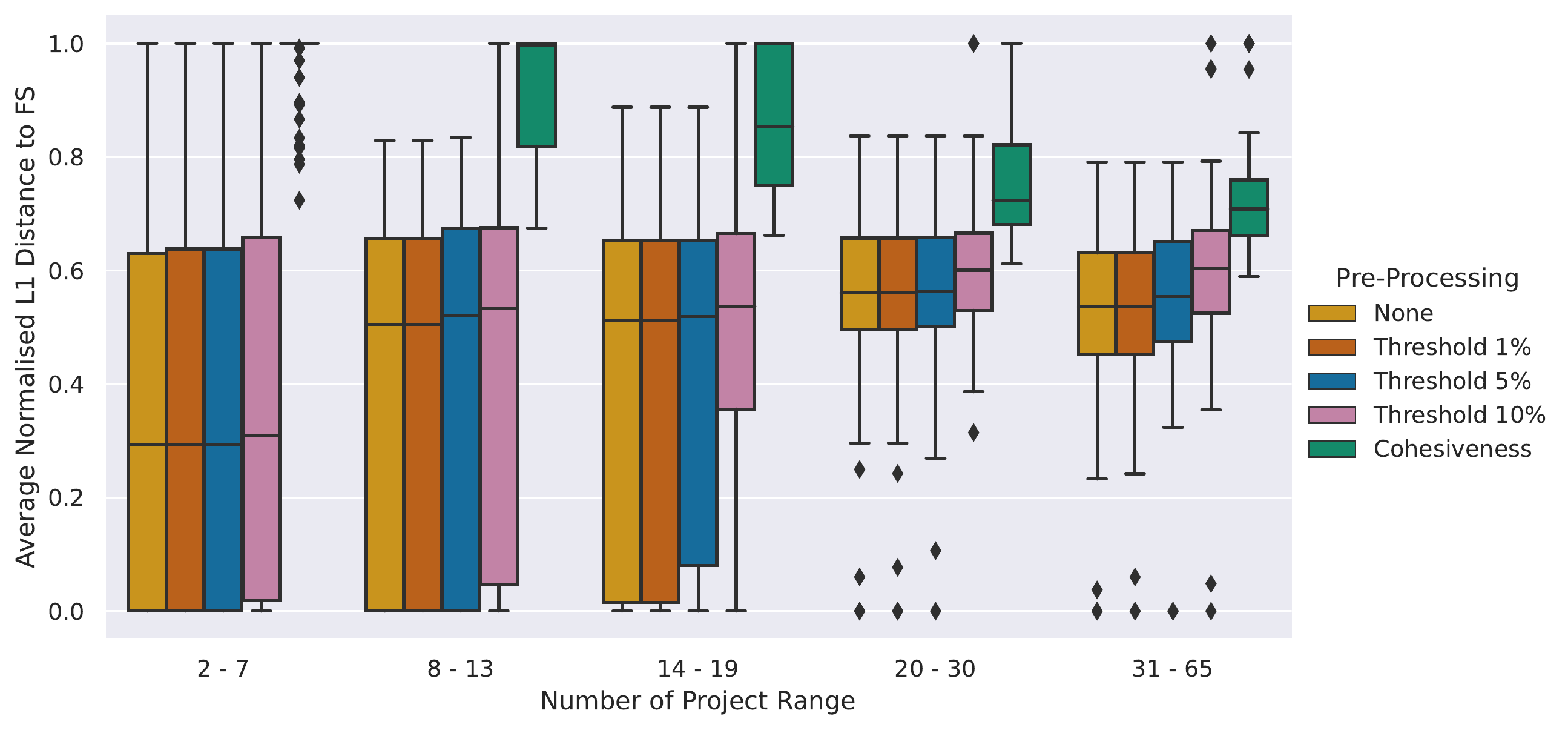}
		
		\vspace*{-5pt}
		\caption{Average capped fair share ratio (left) and $L_1$ distance to FS (right) for Pabulib instances.
			For the latter we actually plot $1 - \nicefrac{1}{n} \cdot \sum_{i \in \agentSet}\frac{|\share(\pi, i) - \fairshare
				(i)|}{\fairshare(i)}$ to obtain a normalised value for which 1 is the best.\textsuperscript{$\star$}
			Each range (for a number of projects) shown on the $x$-axis contains between 60 and 80 instances.}
		\begin{flushleft}
			\small \textsuperscript{$\star$}Note that the empty budget allocation provides an $L_1$ distance to FS of $\fairshare(i)$ for all $i \in \agentSet$. Normalising the $L_1$ distance with $\fairshare(i)$, thus ensures that we display the optimal $L_1$ distance to FS achieved with respect to the worst case.
		\end{flushleft}
		\vspace{1em}
		\label{fig:approx_fs}
	\end{figure*}
	
	\begin{figure*}
		\centering
		\includegraphics[width=\linewidth]{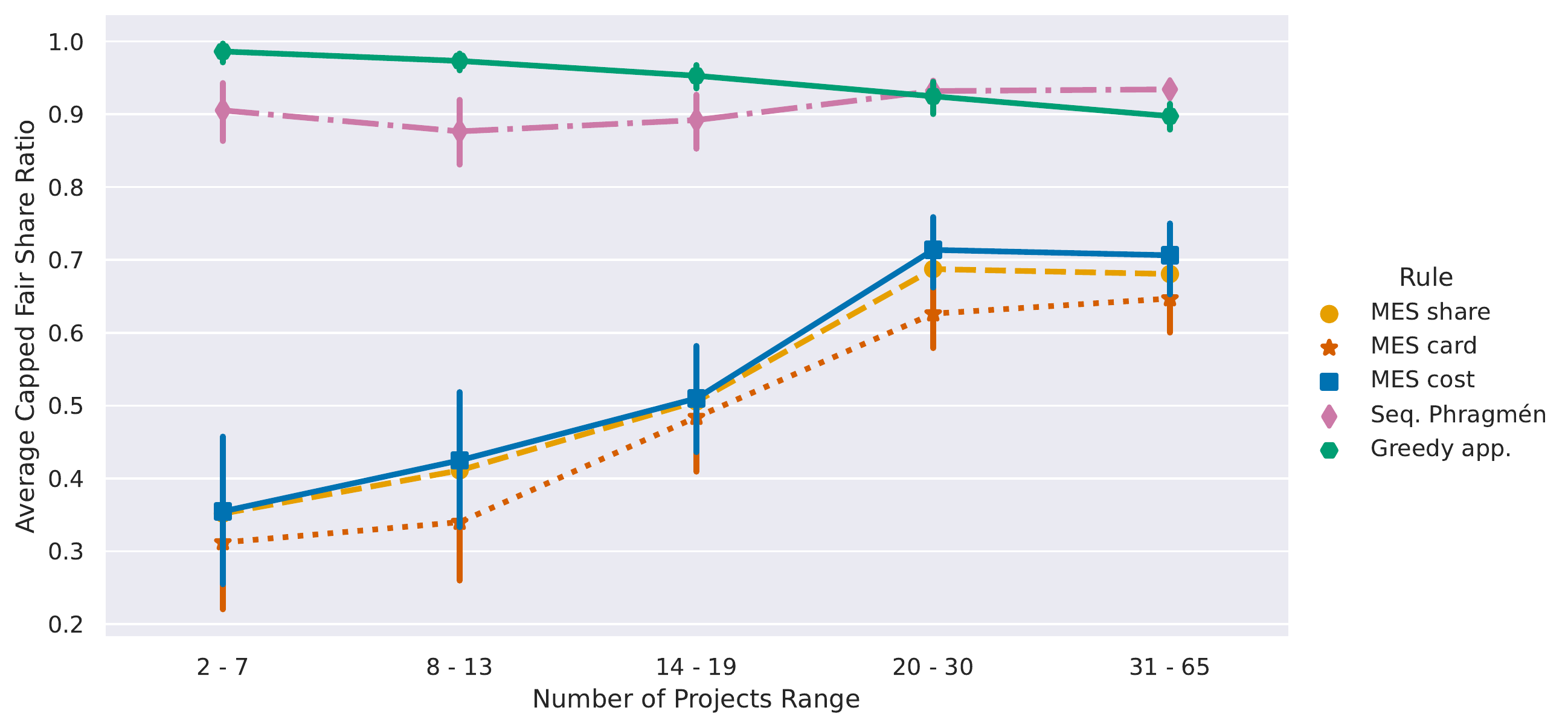}
		\hfill
		\includegraphics[width=\linewidth]{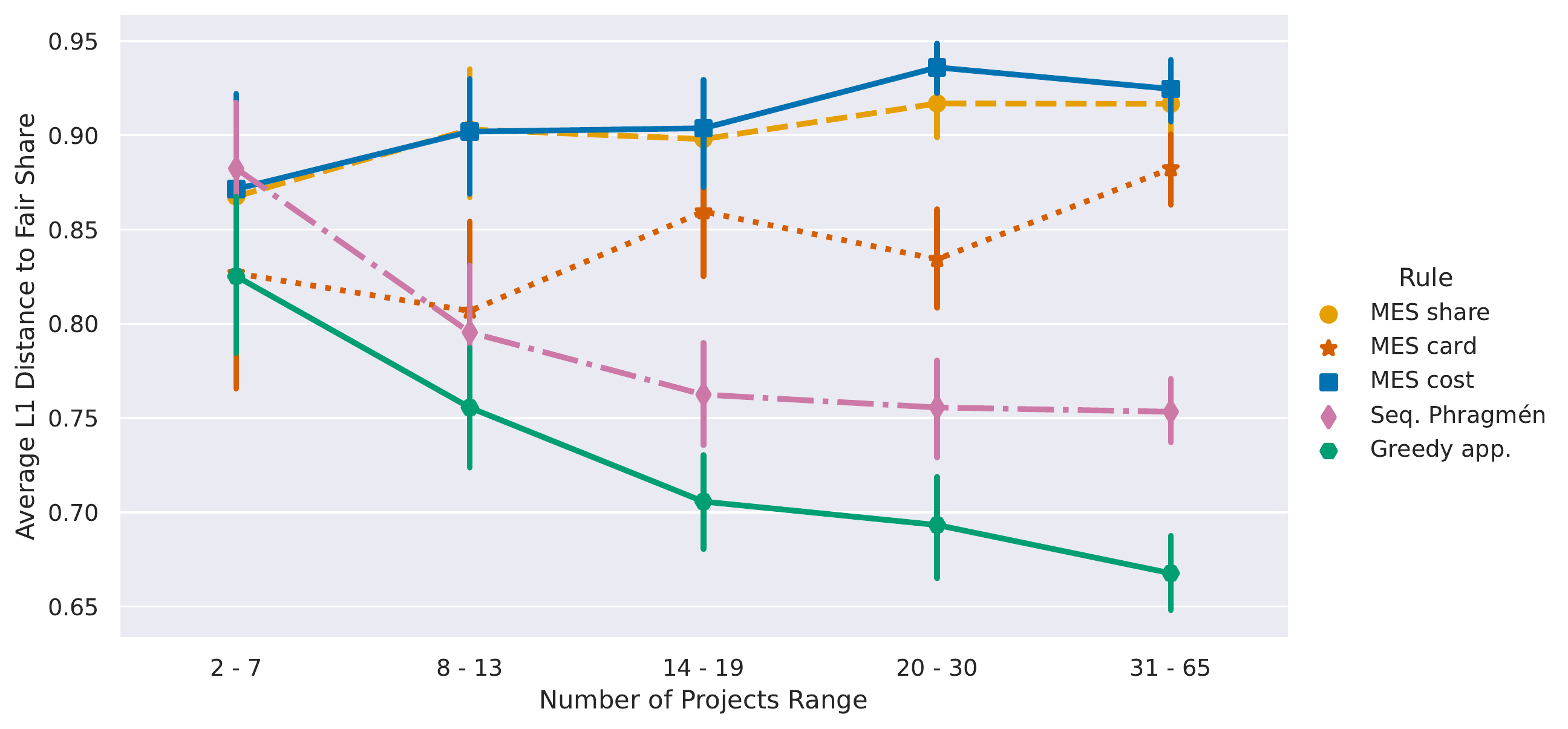}
		\vspace*{-5pt}
		\caption{Average capped fair share ratio (left) and average $L_1$ distance to FS (right) for different rules
			on Pabulib instances. Results are normalised by the optimum value achievable in each instance, giving a score
			between 0 and 1 where 1 is the best.}
		\vspace{1em}
		\label{fig:distance_fs_rules}
	\end{figure*}
	
	\subsection{Optimal Distance to Fair Share}
	
	We propose two ways to measure how close to FS a given budget allocation is.
	The first one is the \emph{average capped fair share ratio}: For every agent~$i$ with approval
	ballot~$A_i$ we divide their actual share by their fair share, capped at~1 in case they get more
	than their fair share, and take the average of this ratio over all agents:
	\[\frac{1}{n} \cdot \sum_{i \in \agentSet} \min \left( \frac{\share(\pi, i)}{\fairshare(i)}, 1 \right). \]
	Our second measure is the \emph{average $L_1$ distance to FS}, measuring, for every agent $i$, the absolute
	value of the difference between their actual share and their fair share:
	\[\frac{1}{n} \cdot \sum_{i \in \agentSet} |\share(\pi, i) - \fairshare(i)|. \]
		
	\noindent
	For each PB instance we computed via \emph{integer linear programs} budget allocations yielding the optimal average
	capped fair share ratio and $L_1$ distance to FS.
	Moreover, to better understand what might cause an instance not to admit a good solution, we also considered
	different ways of preprocessing the instances by removing ``problematic'' projects:
	\begin{itemize}[itemsep=0pt]
		\item Threshold: We remove any project that is not approved by at least~$x$\% of agents. We considered $x=1\%$, 5\%, and 10\%.
		\item Cohesiveness: We remove any project $p$ such that its supporters do not deserve enough money to buy the project, \textit{i.e.}, such that $\frac{|\{i \in \agentSet \mid p \in A_i\}|}{n} b < c(p)$.
	\end{itemize}
	
	\noindent Threshold preprocessing removes under 10\% of projects for a threshold of 1\%, around 10--20\% for a
	threshold of 5\%, and around 20--30\% for a threshold of 10\%.
	Cohesiveness preprocessing removes between 30\% (for the largest instances) and 70\% of projects
	(for the smallest instances)
	
	Let us now turn to our results, presented in Figure~\ref{fig:approx_fs}.
	We draw the following conclusions.
	Without preprocessing, we can provide agents on average between 45\% (for small instances) and 75\%
	(for larger instances) of their fair share, albeit with a lot of variation.
	Furthermore, we can typically guarantee an $L_1$ distance to FS of 50\% of the worst case distance.
	Interestingly, preprocessing helps when using the cohesiveness condition, but not with the threshold condition.
	Note that we do not wish to advocate preprocessing as a method to make budget decisions in practice. Rather, we use it as
	a way of checking whether the failure to guarantee fair share is due to the specific structure of real-life PB instances
	and whether similar instances `nearby' might be significantly better behaved.
	Our experimental findings suggest that this is not the case, and that guaranteeing fair share simply is very hard
	across a wide range of instances.
	Note that, across all instances, for only one instance---with 3 projects and 198 voters---we were able to satisfy FS.
	
	We also investigated approximations of the average capped fair share ratio.
	Specifically, for a number of different given approximation ratios $\alpha\in(0,1]$, we replaced the fair share by
	$\alpha \cdot \fairshare(i)$ in the definition.
	Results indicates that moving from $\alpha = 1$ to $\alpha = 0.2$, has a very small effect on the optimum value
	(around 10\% better for $\alpha = 0.2$).
	We also interpret this result as stating that FS is structurally hard to satisfy.

	\subsection{Distance to Fair Share of Common PB Rules}
	
	We now turn to our second experiment: how close to fair share are the outputs of known PB rules in practice.
	We will consider the following rules: \rulex{}, MES$_{\mathit{card}}$~\citep{PPS21NeurIPS},
	MES$_{\mathit{cost}}$~\citep{PPS21NeurIPS},\footnote{We write MES$_{\mathit{card}}$ and 	MES$_{\mathit{cost}}$ for the rule MES~\citep{PPS21NeurIPS} used with utility functions $u_i(p) := 1$ and $u_i(p) := c(p)$ for all $i \in \agentSet$ and $p \in \projSet$, respectively.} sequential Phragmén~\citep{LCG22}, and greedy approval~\citep{AzSh20}.
	The definitions of the rules here are given in the appendix.
	
	For every PB instance and every PB rule, we compute the outcome returned by the rule and assess how close to the
	optimal value it is in terms of both the average capped fair share ratio and average $L_1$ distance to FS.
	Results are presented in Figure~\ref{fig:distance_fs_rules}.
	
	The first striking observation is that greedy approval is performing extremely well under the capped fair share
	ratio measure.
	This is particularly surprising given how oblivious to the structure of the profile greedy approval is.
	We postulate that this result is due to the high difference in the percentage of the budget used by the different rules: MES rules use
	around 40\% of the budget on average, while greedy approval and sequential Phragmén use around 90\% of the budget.
	Since using more budget can only improve the average capped fair share ratio, this is the most likely explanation for the good performance of greedy approval compared to MES.
	There are no standard ways to extend MES budget allocation in the literature (for PB). It is thus hard
	to compare rules based on the average capped fair share ratio the achieve.
	
	Interestingly, the average $L_1$ distance to FS does not suffer this drawback.
	Indeed, since it also penalises rules that provide agents more than their fair share, spending more is not always
	better.
	Interpreting the results of Figure~\ref{fig:distance_fs_rules} in this light, we conclude that MES rules
	perform better than sequential Phragmén in terms of equality of resources.
	Interestingly, MES$_{\mathit{cost}}$ performs slightly better than \rulex{}.
	It thus provides both good experimental results in terms of equality of resources and
	strong representation guarantees~\citep{PPS21NeurIPS}.
	
	\section{Conclusion}
	\label{sec:conclusion}
	
	In this paper, we have proposed to evaluate the fairness of a participatory budgeting decision
	by using the share of a voter as a measure of the \emph{resources} spent on
	satisfying the needs of voters rather than using the (assumed) \emph{satisfaction}
	each voter might derive from an allocation. Our results suggest that this is an
	interesting measure of fairness that deserves further attention.
	
	In summary, we have seen that perfect fairness in the sense of fair share
	is not always achievable, as is usually the case
	in PB, due to the discrete nature of the process. More surprisingly, our experiments show that,
	in practice, it is often even impossible to achieve outcomes that are close to fair share.
	Nevertheless, we were able to relax the requirements of fair share to define
	several desiderata that can always be satisfied. Using these criteria, we are able to identify 
	\rulex{} as the polynomial-time-computable PB rule that is most equitable in terms
	of resources, as it satisfies both Local-FS as well as EJS-1.
	This result is strengthened by our experimental evaluation that shows that \rulex{}
	selects bundles that are close to optimal with regards to distance to FS.
	
	It is worth noting, however, that MES$_{\mathit{cost}}$
	performs slightly better than \rulex{} in our experiments, hinting at an interesting
	connection between share- and satisfaction-based fairness notions.
	Exploring whether a meaningful compromise
	between these two types of fairness can be achieved is important future work,
	even though, such a compromise would have to be significantly weaker than EJS and EJR
	(in view of, \textit{e.g.}, Example~\ref{EJS-EJR}).
	
	Another important open question is, whether a (natural) PB rule exists
	that satisfies EJS as well as FS and FS-1 whenever they are satisfiable.
	Such a rule would necessarily be intractable, but it would 
	provide strong axiomatic fairness guarantees for PB instances which are small enough
	to allow for the computation of rules with exponential runtime.
	
	\paragraph{Acknowledgments.} We would like to thank the anonymous reviewers, both at this conference and at MPREF-2022, for their valuable and insightful comments.	
	This research was funded in part by the Austrian Science Fund (FWF) under grant numbers
	P3189 and J4581 as well as the Dutch Research Council (NWO) under grant number 639.023.811.
		
	\bibliographystyle{ACM-Reference-Format}
	\bibliography{bib}
	
	\appendix

	\clearpage 

	\section{Full Proofs}
	\subsection{Proof of Proposition~\ref{prop:FS-NP}}
	
	\begin{proof}
		It is clear that checking if a fair share allocation exists is in \complexNP.
		We show \complexNP-hardness by a reduction from \textsc{3-Set-Cover} \citep{FurerY11}, which is the problem of deciding, given a universe $U = \{u_1,\ldots,u_{|U|}\}$, a set $S$ of 3-element subsets of $U$,
		and an integer $k$, whether there exists a subset $S'$ of $S$ such that $\bigcup S' = U$ and $|S|\leq k$.
		Let $(U,S,k)$ be an instance of \textsc{3-Set-Cover}.
		We can assume without loss of generality that $k \leq |U|$.
		
		We build a PB instance as follows: For every element in $U$ there is a voter $1, \dots, |U|$.
		Moreover, there are $2|U|+3$ many auxiliary voters $|U| + 1, \dots, 3|U|+3$ , \textit{i.e.}, $\agentSet = \{1, \dots, 3|U|+3\}$.
		Furthermore, $\projSet = \{p_1, \dots, p_{|S|},p^*\}$, \textit{i.e.}, for every set $s_j \in S$ there is a project
		$p_j$, and there is one auxiliary project $p^*$. We assume unit costs
		and $b = k+1$. The ballot for each voter $i\leq |U|$ is given by $p_j \in A_i$
		if and only if $u_i \in s_j$, \textit{i.e.}, $i$ approves the project representing the set $s_j$ if and only if $u_i$
		is in $s_j$. The auxiliary voters all approve only of $p^*$.
		We claim that there is an allocation $\pi$ that satisfies FS
		if and only if $(U,S,k)$ is a positive instance of \textsc{3-Set-Cover}.
		
		Assume first that there is no set cover of size $k$, so for any set $S'\subseteq S$ of size $k$
		there is an element $u_i$ that is not contained in any set in $S'$.
		It follows that for every allocation $\pi$ of $k$ or fewer projects there is one voter $i \leq |U|$
		with $\share(\pi,i) = 0$. Moreover, for any allocation that does not contain $p^*$, all voters $i > |U|$
		have share $0$. Hence, no allocation with at most $k+1$ projects can satisfy~FS.
		
		Now assume that $S'$ is a set cover of size $k$. We claim that $\pi := \{p_j \mid s_j \in S'\} \cup \{p^*\}$
		satisfies FS. By assumption,
		\[\frac{b}{|\agentSet|} = \frac{k+1}{3|U|+3}\leq \frac{|U|+1}{3|U|+3} =\frac{1}{3}.\]
		Moreover, for every project $p_j$ we have $|\{i \mid p_j \in A_i\}| = 3$ because $|s_j| = 3$ for
		all $s_j \in S$. Now, because $S'$ is a set cover, for each voter $i \leq |U|$ there is a project
		$p_j \in \pi$ such that $p_j \in A_i$. It follows that $\share(\pi,i) \geq \nicefrac{1}{3}$
		for all $i \leq |U|$. For every $i > |U|$ we have $A_i = \{p^*\}$. Now as $p^* \in \pi$
		we have $\share(\pi,i) = \share(A_i,i)$ for all $i > |U|$. It follows that $\pi$ satisfies~FS.
	\end{proof}

	\subsection{Proof of Proposition~\ref{prop:FS1_NP_complete}}
	\begin{proof}
		Checking whether an FS-1 allocation exists clearly is in \complexNP. We show \complexNP-hardness
		using \textsc{3-Set-Cover}~\citep{FurerY11}.
		Consider an instance of the \textsc{3-Set-Cover} problem, \textit{i.e.}, a universe $U = \{u_1, \ldots, u_{|U|}\}$, a set $S = \{s_1, \ldots, s_{|S|}\}$ of 3-element subsets of $U$, and an integer $k \in \mathbb{N}$. Without loss of generality, we make two assumptions: $(i)$ every $u_i \in U$ appears in at least one $s_j \in S$; $(ii)$ $\nicefrac{|U|}{3} \leq k < |U|$. Note that whenever one of these assumptions is violated, the answer of the \textsc{3-Set-Cover} problem can easily be found in polynomial time.
		
		We furthermore make the assumption that $k > \nicefrac{|U|}{3}$. This assumption is justified by the following reduction: Consider an instance $\tuple{U^1, S^1, k^1}$ for which $k^1 = \nicefrac{|U^1|}{3}$. We reduce this to an instance $\tuple{U^2, S^2, k^2}$ for which $k^2 > \nicefrac{|U^2|}{3}$ holds. This instance is constructed as follows: $U^2 = U^1 \cup \{u^2_1, u^2_2, u^2_3, u^2_4\}$, $S^2 = S^1 \cup \{\{u^2_1, u^2_2, u^2_3\}, \{u^2_2, u^2_3, u^2_4\}\}$, and $k^2 = k^1 + 2$. It is clear that a subset $S'_1 \subseteq S^1$ is a solution of the \textsc{3-Set-Cover} problem for $\tuple{U^1, S^1, k^1}$ if and only if $S'_2 = S'_1 \cup \{\{u^2_1, u^2_2, u^2_3\}, \{u^2_2, u^2_3, u^2_4\}\}$ is a suitable solution for $\tuple{U^2, S^2, k^2}$. Moreover, note $k^2 = \nicefrac{|U^1|}{3} + 2$ and that $\nicefrac{|U^2|}{3} = \nicefrac{|U^1|}{3} + \nicefrac{2}{3}$. We thus have $k^2 > \nicefrac{|U^2|}{3}$. This shows that \textsc{3-Set-Cover} remains \complexNP-complete even if we require that $k > \nicefrac{|U|}{3}$.
		
		Let us now get to the reduction. We distinguish between two cases based the relative value of $k$ and $|U|$, and construct different instances in both cases.
		
		\underline{\smash{Whenever $\nicefrac{1}{3}|U| < k \leq \nicefrac{2}{3}|U|$}}, we construct a PB instance $I$ as follows. The set of voters is $\agentSet = \{1, \ldots, 2|U|\}$, \textit{i.e.}, there are two voters per element of $U$. The set of projects is $\projSet = \{p_j^1 \mid s_j \in S\} \cup \{p_j^2 \mid s_j \in S\}$, \textit{i.e.}, there are two projects per element of $S$. All projects have cost 1 and the budget limit is $b = k$. The approval ballots are such that $A_i = \{p_j^1 \mid u_i \in s_j\} \cup \{p_j^2 \mid u_i \in s_j\}$ for all $1 \leq i \leq |U|$, and $A_i = \{p_j^1 \mid u_{i - |U|} \in s_j\} \cup \{p_j^2 \mid u_{i - |U|} \in s_j\}$ for all $|U| + 1 \leq i \leq 2|U|$, \textit{i.e.}, agent $i$ approves of the two projects representing the set $s_j$ if and only if $u_i \in s_j$. We know want to prove that there exists a suitable $S' \subseteq S$ to answer the \textsc{3-Set-Cover} problem, if and only if, there exists an FS-1 budget allocation in the PB instance previously described.
		
		Observe that in $I$ and $\profile$, to reach their fair share, every agent needs to have two projects from the ones they approve of selected. Indeed, remember that we assumed $\nicefrac{1}{3}|U| < k \leq \nicefrac{2}{3}|U|$. Since $\nicefrac{b}{n} = \nicefrac{k}{2|U|}$, we thus have
		$\nicefrac{1}{6} < \nicefrac{b}{n} \leq \nicefrac{1}{3}$.
		
		Moreover, for every project $p \in \projSet$, there are exactly six agents approving of it, so for any agent $i \in \agentSet$ approving of $p$, we have $\share(\{p\}, i) = \nicefrac{1}{6}$. Given that every agent $i \in \agentSet$ approves of at least 2 projects, we thus have $\share(A_i, i) \geq \nicefrac{1}{3}$. Overall, we know that for every $i \in \agentSet$, $\nicefrac{1}{6} < \fairshare(i) \leq \nicefrac{1}{3}$, so every agent needs two projects to reach their fair share.
		
		Now, an allocation $\pi \in \allocSet(I)$ satisfies FS-1 if and only if every agent has a non-zero share in $\pi$: According to the above, if an agent has a non-zero share then adding an extra project will always grant them their fair share; and if an agent has a zero share, one would need to add two extra projects to $\pi$ which is not allowed. This is possible if and only if there exists a set $S' \subseteq S$ of size at most $k$ such that every elements of $U$ appears in at least one element of $S'$. This concludes the proof for this case.
		
		\underline{\smash{Whenever $\nicefrac{2}{3}|U| < k \leq |U|$}}, we construct a PB instance $I$ as follows. The set of voters is $\agentSet = \{1, \ldots, 2|U|\} \cup \{2|U| + 1, \ldots, 3k\}$, \textit{i.e.}, there are two voters per element of $U$, together with a certain number of additional voters. The set of projects is $\projSet = \{p_j^1 \mid s_j \in S\} \cup \{p_j^2 \mid s_j \in S\} \cup \{p^\star\}$, \textit{i.e.}, there are two projects per element of $S$, and an additional one. All projects have cost 1 and the budget limit is $b = k$. The approval ballots are such that$A_i = \{p_j^1 \mid u_i \in s_j\} \cup \{p_j^2 \mid u_i \in s_j\}$ for all $1 \leq i \leq |U|$, $A_i = \{p_j^1 \mid u_{i - |U|} \in s_j\} \cup \{p_j^2 \mid u_{i - |U|} \in s_j\}$ for all $|U| + 1 \leq i \leq 2|U|$, and $A_i = \{p^\star\}$ for all $2|U| < i \leq 3k$. In total, agent $i$, for $i \leq 2|U|$, approves of the two projects representing the set $s_j$ if and only if $u_i \in s_j$; while the additional agents all approve only of $p^\star$.
		
		Let us discuss few facts about this construction. Remember that $\nicefrac{2}{3}|U| < k \leq |U|$. This implies $3k > 2|U|$, \textit{i.e.}, $3k \geq 2|U| + 1$ since $k \in \mathbb{N}$. The set of agent $\agentSet$ is thus well defined. Moreover, by construction we have $\nicefrac{b}{n} = \nicefrac{k}{3k} = \nicefrac{1}{3}$. As for the previous case, for every agent $i \leq 2|U|$, the marginal share of every project they approve of is $\nicefrac{1}{6}$. They thus deserve a share of at least $\nicefrac{1}{3}$ to get their fair share, which can only be done by selecting at least two projects. Observe in addition that every budget allocation $\pi \in \allocSet(I)$ satisfies the FS-1 condition for agents $2|U| + 1 \leq i \leq 3k$ (the project $p^\star$ can always be added, if it is not already in $\pi$, by virtue of FS-1).
		
		Overall, an allocation $\pi \in \allocSet(I)$ satisfies FS-1 if and only if every agent $i \leq 2|U|$ has a non-zero share in $\pi$. Such a $\pi$ exists if and only if there is a set $S' \subseteq S$ of size at most $k$ such that every elements of $U$ appears in at least one element of $S'$. This concludes the proof.
	\end{proof}
	
	\subsection{Proof of Proposition~\ref{Prop:Sat-EJS}}
	
	\begin{algorithm}[t]
		\DontPrintSemicolon
		
		\caption{Greedy EJS}
		\label{algo:greedyEJS}
		
		\KwIn{An instance $I = \tuple{\projSet, c, b}$ and a profile $\profile$}
		\KwOut{A budget allocation $\pi \in \allocSet(I)$ satisfying EJS}
		
		Intialise $\pi$ and $N^\star$ as the empty set: $\pi \gets \emptyset$, $N^\star \gets \emptyset$\;
		\While{there exists an $N \subseteq \agentSet \setminus N^\star$ with $N \neq \emptyset$ and a $P \subseteq \projSet \setminus \pi$ with $P \neq \emptyset$, such that $N$ is $P$-cohesive}{
			Let $N \subseteq \agentSet \setminus N^\star$ and $P \subseteq \projSet \setminus \pi$ be such that:
			\begin{align*}
				(N, P) \in \argmax_{\substack{(N', P') \ \in\  2^{\agentSet \setminus N^\star} \times 2^{\projSet \setminus \pi} \\ N' \text{ is }P'\text{-cohesive}}} \ \max_{i \in N'} \ \share(P', i)
			\end{align*}\;
			\vspace{-1em} Select the projects in $P$: $\pi \gets \pi \cup P$\;
			Agents in $N$ have been satisfied: $N^\star \gets N^\star \cup N$\;
		}
		\KwRet{the budget allocation $\pi$}
		
	\end{algorithm}
	
	\noindent
	
	\begin{proof}
		We show that Algorithm \ref{algo:greedyEJS} computes a feasible budget allocation that satisfies EJS. Let us consider an arbitrary instance $I = \tuple{\projSet, c, b}$ and profile $\profile$.
		
		We first show that the budget allocation returned by the algorithm indeed is feasible.
		
		\begin{claim}
			The budget allocation $\pi$ returned by Algorithm \ref{algo:greedyEJS} on $I$ and $\profile$ is feasible.
		\end{claim}
		
		\begin{claimproof}
			Consider the run of the algorithm on $I$ and $\profile$ and assume that the while-loop is run $k$ times. Let us call $(N_j, P_j)$ the sets of agents and projects that are selected during the $j$-ith run of the while-loop, for all $j \in \{1, \ldots, k\}$. We then have:
			\[c(\pi) = \sum_{j = 1}^k c(P_j) \leq \sum_{j = 1}^k \frac{|N_j| \times b}{n} \leq b.\]
			The first equality comes from the fact that $P_1, \ldots, P_k$ is a partition of $\pi$. The inequality is derived from the fact that $N_j$ is a $P_j$-cohesive group, for all $j \in \{1, \ldots, k\}$ (it is an inequality because for any of the projects $p \in P_j$, some agents outside of $N_j$ may approve of it; $c(p)$ can thus be split among more than $|N_j|$ agents). The final inequality is linked to the fact all the $N_1, \ldots, N_k$ are pairwise disjoint. Overall, the outcome of Algorithm \ref{algo:greedyEJS} is a feasible budget allocation.
		\end{claimproof}
		
		\medskip\noindent Let us now prove that the algorithm does compute an EJS budget allocation.
		
		\begin{claim}
			The budget allocation $\pi$ returned by Algorithm \ref{algo:greedyEJS} on $I$ and $\profile$ satisfies EJS.
		\end{claim}
		
		\begin{claimproof}
			Assume towards a contradiction that $\pi$ violates EJS. Then, there must exist some $N \subseteq \agentSet$ and $P \subseteq \projSet$ such that $N$ is $P$-cohesive but also such that, for all agents $i \in N$, we have $\share(\pi, i) < \share(P, i)$. Note that, if $P \nsubseteq \pi$, this means that at the end of the algorithm either one agent $i \in N$ has been satisfied ($i \in N^\star$ when the algorithm returns) or that one project $p \in P$ has been selected ($p \in \pi$ when the algorithm returns). We distinguish these two cases.
			
			First, consider the case where one agent has been satisfied by the end of the algorithm. Using the same notation as for the previous claim, there exists then a smallest $j \in \{1, \ldots, k\}$ such that there exist $i^\star \in N \cap N_j$. Given that $(N, P)$ has not been selecting during that run of the while loop, it means that:
			\[\max_{i' \in N_j} \share(P_j, i') \geq \max_{i \in N} \share(P, i).\]
			Since the cost of a project is split equality among its supporters, it is easy to observe that for any $P$-cohesive group $N$, and for every two agents $i, i' \in N$, we have $\share(P, i) = \share(P, i')$. Moreover, we also have that $\share(P', i) \leq \share(P, i)$ for any $P' \subseteq P \subseteq \projSet$ and $i \in \agentSet$. Overall, for our specific agent $i^\star \in N \cap N_j$, we have:
			\begin{align*}
				\share(\pi, i^\star) & \geq \max_{i' \in N_j} \share(P_j, i') \\
				& \geq \max_{i \in N} \share(P, i) \\
				& \geq \share(P, i^\star),
			\end{align*}
			which contradicts the fact that $\pi$ fails EJS.
			
			Let us now consider the second case, \textit{i.e.}, when $P \cap \pi \neq \emptyset$ but $P \nsubseteq \pi$. In this case, it is important to see that if $N$ is $P$-cohesive, then it is also $P'$-cohesive for all $P' \subseteq P$. Then, we can run the same proof considering the $(P \setminus \pi)$-cohesive group $N$. Iterating this argument, would either lead to the conclusion that $P \subseteq \pi$, a contradiction, or to another contradiction due to the first case we considered (when some agent of $N$ is already satisfied).
		\end{claimproof}
		
		\medskip\noindent Finally, it should be noted that Algorithm \ref{algo:greedyEJS} always terminates. Indeed, after each run of the while-loop, at least one agent is added to the set $N^\star$. Moreover, if $N^\star = \agentSet$, the condition of the while-loop would be violated and the algorithm would terminate. Overall at most $n$ runs through the while-loop can occur. This concludes the proof.
	\end{proof}
	
	\subsection{Proof of Theorem~\ref{Thm:EJS-hard}}
	
	\begin{proof}
		Assume, that there is an algorithm $\mathbb{A}$ that always computes an allocation satisfying EJS.

		We will make use of the \textsc{Subset-Sum} problem, known to be \complexNP-hard. In this problem, we are given as input a set $S = \{s_1, \ldots, s_m\}$ of integers and a target $t \in \mathbb{N}$ and we wonder whether there exists an $X \subseteq S$ such that $\sum_{x \in X} x = t$.
		
		Given $S$ and $t$ as described above, we construct $I = \tuple{\projSet, c, b}$ and $\profile$ as follows. We have $m$ projects $\projSet = \{p_1, \ldots, p_m\}$ with the following cost function $c(p_j) = s_j$ for all $j \in \{1, \ldots, m\}$ and a budget limit $b = t$. There is moreover only one agent, who approves of all the projects.
		
		Now, $(S,t)$ is a positive instance of \textsc{Subset-Sum} if and only if there
		is a budget allocation $\pi \in \allocSet(I)$ that cost is exactly $b$. If such an allocation $\pi$
		exists, then the one voter $1$ is $\pi$-cohesive. Therefore, any allocation $\pi'$ that satisfies
		EJS must give that voter $\share(1, \pi') \geq \share(1,\pi) = c(\pi)$.
		Hence, $(S,t)$ is a positive instance of \textsc{Subset-Sum} if and only if
		$c(\mathbb A(I,\profile)) = b$. This way, we can
		use $\mathbb{A}$ to solve \textsc{Subset-Sum} in polynomial time.
	\end{proof}
	
	\subsection{Proof of Proposition~\ref{prop:runningTimeGreedyEJS}}

	\begin{proof}
		Consider an instance $I = \tuple{\projSet, c, b}$ and a profile $\profile$ on which Algorithm \ref{algo:greedyEJS} is run.
		
		The first thing to note is that at least one agent is added to $N^\star$ during each run through the while-loop and that, if ever $N^\star = \agentSet$, the condition of the while-loop is trivially satisfied. Overall, the while-loop can only be run $n$ times.
		
		Let us have a closer look at what is happening inside the while-loop. The main computational task here is the maximisation that goes through all subsets of agents and of projects. We will show that we can avoid going through all subsets of agents. Indeed, consider a subset of projects $P \subseteq \projSet$ and let $N \subseteq \agentSet$ be the largest set of agents such that for all $i \in N$, $P \subseteq A_i$. Note such a set $N$ can be efficiently computed (by going through all the approval ballots). Now, if some group of agents is $P$-cohesive, then for sure $N$ should be $P$-cohesive. Moreover, note that for any $P$-cohesive group $N'$, and for every two agents $i \in N$ and $i' \in N \cup N'$, we have $\share(P, i) = \share(P, i')$. Overall, one can, without loss of generality, only consider the group of agents $N$ when considering the subset of projects $P$. This implies that the maximisation step can be computed by going through all the subsets of projects and, for each of them, only considering a single subset of agents (that is efficiently computable).
	\end{proof}

	\subsection{Proof of Theorem~\ref{Thm:Rule-X-Local-EJS}}
	
	\begin{proof}
		
		Let $\pi = \{p_1, \dots, p_k\}$ be the budget allocation output by \rulex{} on instance $I$ and profile $\profile$,
		where $p_1$ was selected first, $p_2$ second etc. For any $1 \leq j \leq k$, set $\pi_j := \{p_1, \ldots, p_j\}$. Consider $N \subseteq \agentSet$, a $P$-cohesive group, for some $P \subseteq \projSet$.
		We show that $\pi$ satisfies Local-EJS for $N$.
		If $P \subseteq \pi$ then Local-EJS is satisfied by definition.
		We will thus assume that $P \not \subseteq \pi$.
		
		Let $k^*$ be the first round after which there exists a voter $i^* \in N$ whose load
		is larger than $\nicefrac{b}{n} - \nicefrac{1}{|N|}$. Such a round must exist
		as otherwise the voters in $N$ could afford another project from $P$.
		As we assumed $P \not \subseteq \pi$, this would mean that \rulex{} cannot have stopped.
		Let $\pi^* = \pi_{k^*}$ and consider an arbitrary project $p^* \in P \setminus \pi^*$.
		Our goal is to prove that $\pi^*$ satisfies Local-EJS for $N$, \textit{i.e.}:
		\begin{align}
			& \share(\pi^* \cup \{p^*\}, i^*) > \share(P, i^*) \nonumber \\
			\Leftrightarrow \quad & \share(\pi^*, i^*) > \share(P \setminus \{p^*\}, i^*) \nonumber \\
			\Leftrightarrow \quad & \share(\pi^* \cap P, i^*) + \share(\pi^* \setminus P, i^*) > \nonumber \\
			& ~~\share(P \cap \pi^*, i^*) + \share(P \setminus (\pi^* \cup \{p^*\}), i^*) \nonumber \\
			\Leftrightarrow \quad & share(\pi^* \setminus P, i^*) > \share(P \setminus (\pi^* \cup \{p^*\}), i^*). \label{eq:RuleXLocalEJS_mainClaim}
		\end{align}
		We will now work on each side of inequality \eqref{eq:RuleXLocalEJS_mainClaim} to eventually prove that it is indeed satisfied.
		
		\medskip\noindent	
		We start by the left-hand side of \eqref{eq:RuleXLocalEJS_mainClaim}. Let us first introduce some notation that will allow us to eason in terms of share per unit of load. For a project $p \in \pi$, we denote by $\alpha(p)$ the smallest $\alpha \in \mathbb{R}_{> 0}$ such that $p$ was $\alpha$-affordable when \rulex{} selected it.
		Moreover, we define $q(p)$---the share that a voter that contributes fully to $p$
		gets per unit of load---as $q(p) := \nicefrac{1}{\alpha(p)}$.
		
		Since before round $k^*$, agent $i^*$ contributed in full for all projects in $\pi^*$ (as $\ell_{i^*} < \nicefrac{b}{|N|}$ after each round $1, \dots, k^*$),
		we know that $\alpha(p) \cdot \share(\{p\}, i^*)$
		equals the contribution of $i^*$ for $p$, and that for any $p \in \pi^*$. We thus have:
		\begin{align}
			& \share(\pi^* \setminus P, i^*) \nonumber \\
			= \quad & \sum_{p \in \pi^* \setminus P} \share(\{p\}, i^*) \nonumber \\
			= \quad & \sum_{p \in \pi^* \setminus P} \alpha(p) \cdot \share(\{p\}, i^*) \cdot \frac{1}{\alpha(p)} \nonumber \\
			= \quad & \sum_{p \in \pi^* \setminus P} \gamma_{i^*}(p) \cdot q(p), \label{eq:RuleXLocalEJS_LeftHand1}
		\end{align}
		where $\gamma_{i^*}(p)$ denotes the contribution of $i^*$ to any $p \in \pi$, defined such that if $p$ has been selected at round $j$, \textit{i.e.}, $p = p_j$, then
		$\gamma_{i^*}(p) = \gamma_{i^*}(\pi_j, \alpha(p_j), p_j)$.
		
		Now, let us denote by $q_{\min}$ the smallest $q(p)$ for any $p \in \pi^* \setminus P$. From \eqref{eq:RuleXLocalEJS_LeftHand1}, we get:
		\begin{equation}
			\share(\pi^* \setminus P, i^*) \geq q_{\min} \sum_{p \in \pi^* \setminus P} \gamma_{i^*}(p). \label{eq:RuleXLocalEJS_LeftHand2}
		\end{equation}
		
		\noindent
		We now turn to the right-hand side of \eqref{eq:RuleXLocalEJS_mainClaim}. We introduce some additional notation for that. For every project $p \in P$, we denote by $q^*(p)$ the share per load that a voter
		in $N$ receives if only voters in $N$ contribute to $p$, and they all contribute in full to $p$, defined as:
		\[q^*(p) = \frac{\share(\{p\}, i)}{\nicefrac{1}{|N|}} = \frac{|N|}{|\{A \in \profile \mid p \in A\}|},\]
		where $i$ is any agent in $N$.
		
		We have then:
		\begin{align}
			& \share(P \setminus (\pi^* \cup \{p^*\}), i^*) \nonumber \\
			= \quad & \sum_{p \in P \setminus (\pi^* \cup \{p^*\})} \share(\{p\}, i^*) \nonumber \\
			= \quad & \sum_{p \in P \setminus (\pi^* \cup \{p^*\})} \frac{\share(\{p\}, i^*)}{\nicefrac{1}{|N|}} \cdot \frac{1}{|N|} \nonumber \\
			= \quad & \sum_{p \in P \setminus (\pi^* \cup \{p^*\})} q^*(p) \cdot \frac{1}{|N|} \label{eq:RuleXLocalEJS_RightHand1}
		\end{align}
		Setting $q^*_{\max}$ to be the largest $q^*(p)$ for all $p \in P \setminus (\pi^* \cup \{p^*\})$, \eqref{eq:RuleXLocalEJS_RightHand1} gives us:		
		\begin{equation}\share(P \setminus (\pi^* \cup \{p^*\}), i^*) \leq q^*_{\max} \cdot \frac{|P \setminus (\pi^* \cup \{p^*\})|}{|N|}. \label{eq:RuleXLocalEJS_RightHand2}
		\end{equation}
		
		\noindent
		With the aim of proving inequality \eqref{eq:RuleXLocalEJS_mainClaim}, we want to show that
		\begin{equation}
			q_{\min} \cdot \sum_{p \in \pi^* \setminus P} \gamma_{i^*}(p) > q^*_{\max} \cdot \frac{|P \setminus (\pi^* \cup \{p^*\})|}{|N|}.
			\label{eq:RuleXLocalEJS_mainClaim2}
		\end{equation}
		Note that proving that this inequality holds, would in turn prove \eqref{eq:RuleXLocalEJS_mainClaim} thanks to
		\eqref{eq:RuleXLocalEJS_LeftHand2} and \eqref{eq:RuleXLocalEJS_RightHand2}. We divide the proof of \eqref{eq:RuleXLocalEJS_mainClaim2} into two claims.
		
		\begin{claim}
			$q_{\min} \geq q^*_{\max}$.
		\end{claim}
		
		\begin{claimproof}
			Consider any project $p' \in P \setminus (\pi^* \cup \{p^*\})$. It must be the case that $p'$ was at least $\nicefrac{1}{q^*(p)}$-affordable in round $1, \dots, k^*$, for all $p \in \pi^*$,
			as all voters in $N$ could have fully contributed to it based on how we defined $k^*$.
			
			As no $p' \in P \setminus (\pi^* \cup \{p^*\})$ was selected by \rulex{}, we know that all projects
			that have been selected must have been at least as affordable, \textit{i.e.}, for all $p \in \pi^*$ and $p' \in P \setminus (\pi^* \cup \{p^*\})$ we have:
			\begin{align*}
				& \alpha(p) \leq \frac{1}{q^*(p')} \\
				\Leftrightarrow \quad & q(p) \geq q^*(p') \\
				\Leftrightarrow \quad & q_{\min} \geq q^*_{\max}.
			\end{align*}
			This concludes the proof of our first claim.
		\end{claimproof}

		\begin{claim}
			$\displaystyle \sum_{p \in \pi^* \setminus P} \gamma_{i^*}(p) > \frac{|P \setminus (\pi^* \cup \{p^*\})|}{|N|}$.
		\end{claim}
		
		\begin{claimproof}
			From the choice of $k^*$, we know that the load of agent $i^*$ at round $k^*$ is such that:
			\[\ell_{i^*}(\pi^*) + \frac{1}{|N|} > \frac{b}{n}.\]
			On the other hand, since $N$ is a $P$-cohesive group, we know that:
			\[\frac{|P|}{|N|} = \frac{|P \setminus \{p^*\}|}{|N|}+ \frac{1}{|N|} \leq \frac{b}{n}.\]
			Linking these two facts together, we get:
			\[\ell_{i^*}(\pi^*) > \frac{|P \setminus \{p^*\}|}{|N|}.\]
			By the definition of the load, we thus have:
			\[\ell_{i^*}(\pi^*) = \sum_{p_j \in \pi^*} \gamma_{i^*}(p_j) > \frac{|P \setminus \{p^*\}|}{|N|}.\]
			This is equivalent to:
			\begin{multline}
				\sum_{p_j \in P \cap \pi^*} \gamma_{i^*}(p_j) + \sum_{p_j \in P \setminus \pi^*} \gamma_{i^*}(p_j) > \\ \frac{|P \cap \pi^*|}{|N|} + \frac{|P \setminus (\pi^*\cup \{p^*\})|}{|N|} \label{eq:RuleXLocalEJS_SecondClaim}
			\end{multline}
			Now, we observe that every voter in $N$ contributed in full for every
			project in $\pi^*$. It follows that the contribution of every voter in $N$ for a
			project $p_j \in P \cap \pi^*$ is smaller or equal the contribution needed if
			the voters in $N$ would fund the project by themselves. In other words for all
			$p \in P \cap \pi^*$ we have:
			\[\gamma_{i^*}(p) \leq \frac{1}{|N|}.\]
			It follows then that:
			\[\sum_{p_j \in P \cap \pi^*} \gamma_{i^*}(p_j) \leq \frac{|P \cap \pi^*|}{|N|}.\]
			For \eqref{eq:RuleXLocalEJS_SecondClaim} to be satisfied, we must have that:
			\[\sum_{p_j \in \pi^* \setminus P} \gamma_{i^*}(p_j) > \frac{|P \setminus (\pi^*\cup \{p^*\})|}{|N|}\]
			This concludes the proof of our second claim.
		\end{claimproof}
		
		\medskip\noindent Putting together these two claims immediately shows that inequality \eqref{eq:RuleXLocalEJS_mainClaim2} is satisfied, which in turn shows that \eqref{eq:RuleXLocalEJS_mainClaim} also is. Since $P$, $N$ and $p^*$ were chosen arbitrarily, this shows that \rulex{} satisfied Local-EJS in the unit-cost setting.
	\end{proof}
	
	\subsection{Proof of Theorem \ref{thm:share-all-relations}, Claim \eqref{relations-fs1-implies-ejs1}}

	\begin{proof}[Proof of \eqref{relations-fs1-implies-ejs1}.]
		Let $\pi$ satisfy FS-1.
		First, consider an agent $i \in \agentSet$ such that $\share(A_i, i) < \nicefrac{b}{n}$.
		For FS-1 to be satisfied, there must be a project $p \in \projSet$ such that $\share(\pi \cup \{p\}, i) \geq \share(A_i, i)$.
		This entails that $|A_i \setminus \pi| \leq 1$ should be the case. Hence, the conditions for EJS-1 are trivially
	satisfied for agent~$i$.
		
		Consider now an agent $i \in \agentSet$ such that $\share(A_i, i) \geq \nicefrac{b}{n}$.
		Since $\pi$ satisfies FS-1, we know that there must be a project $p \in \projSet$ such that $\share(\pi \cup \{p\}, i) \geq \nicefrac{b}{n}$.
		Let $N \subseteq \agentSet$ be a $P$-cohesive group, for some $P \subseteq \projSet$, such that $i \in N$.
		By definition of a cohesive group, we know that $c(P) \leq \nicefrac{b}{n} \cdot |N|$. Hence, $\share(P, i) \leq \nicefrac{b}{n}$.
		Overall, we get $\share(\pi \cup \{p\}, i) \geq \nicefrac{b}{n} \geq \share(P, i)$ and  $\pi$ satisfies EJS-1.
	\end{proof}

	\subsection{Counterexamples from Section~\ref{sec:taxonomy-share}}
	
	\begin{example}[FS-1 not implying Local-EJS]
		Consider the following instance with four projects, a budget limit of $b = 12$, and three agents.
		\begin{center}\small
			\begin{tabular}{ccccc}
				\toprule
				& $p_1$ & $p_2$ & $p_3$ & $p_4$ \\
				\midrule
				Cost & 4 & 2 & 5 & 7 \\
				\midrule
				$A_1$ & \xmark & \xmark & \xmark & \cmark \\
				$A_2$ & \cmark & \cmark & \cmark & \xmark \\
				$A_2$ & \cmark & \cmark & \cmark & \xmark \\
				\bottomrule
			\end{tabular}
		\end{center}
		Allocation $\pi = \{p_1, p_4\}$ satisfies FS-1 but fails Local-EJS: the $\{p_2, p_3\}$-cohesive group $\{2, 3\}$ deserves a share of 3.5, but adding $p_2$ to $\pi$ would not meet this requirement.
	\end{example}
	
	\begin{example}[Neither Strong-EJS nor FS-1 implying Local-FS]
		Consider the following instance with five projects, a budget limit of $b = 16$, and two agents.
		\begin{center}\small
			\begin{tabular}{cccccc}
				\toprule
				& $p_1$ & $p_2$ & $p_3$ & $p_4$ & $p_5$ \\
				\midrule
				Cost & 12 & 12 & 1 & 1 & 4\\
				\midrule
				$A_1$ & \cmark & \cmark & \cmark & \cmark & \xmark \\
				$A_2$ & \cmark & \xmark & \xmark & \xmark & \cmark \\
				\bottomrule
			\end{tabular}
		\end{center}
		In this instance, the cohesive groups are: $\{1, 2\}$ is $\{p_1\}$-cohesive, $\{1\}$ is $\{p_3\}$-cohesive, $\{p_4\}$-cohesive and $\{p_3, p_4\}$-cohesive, and $\{2\}$ is $\{p_5\}$-cohesive. Overall, to satisfy Strong-EJS, a budget allocation should provide a share of at least 6 to agents 1 and 2. The budget allocation $\pi = \{p_1, p_5\}$ thus satisfies Strong-EJS (note that is is exhaustive). However, one can easily check that $\pi$ does not satisfy the conditions of Local-FS as adding $p_3$ to $\pi$ only provides a share of $7$ to agent 1 while the lower bound for their fair share is $8$.
		
		Note that budget allocation $\pi = \{p_1, p_5\}$ does satisfy FS-1: the fair share of the voters is 8; voter 2 already has a share of 10 in $\pi$, and voter 1 would get a share of 18 in $\pi \cup \{p_2\}$.
	\end{example}
	
	\noindent
	Example~\ref{ex:Local-FS does not imply EJS-1} also shows that neither EJS, Local-EJS, nor EJS-1 imply Local-FS. This completes the picture.
	
	We already know that Local-FS does not imply FS-1 (Local-FS allocation always exists (as a consequence of Theorem~\ref{thm:Rule-X-Local-FS}) while FS-1 is not always satisfiable (Proposition~\ref{prop:FS-1-may-not-exist})).
	The following counterexample shows that even assuming the existence of an FS-1 allocation, Local-FS does not imply FS-1.
	
	\begin{example}[Local-FS not implying FS-1, even when an FS-1 allocation exists]
		Consider the following instance with three projects, a budget limit of $b = 10$, and four agents.
		\begin{center}\small
			\begin{tabular}{cccc}
				\toprule
				& $p_1$ & $p_2$ & $p_3$  \\
				\midrule
				Cost & 3 & 4 & 7 \\
				\midrule
				$A_1$ & \xmark & \xmark & \cmark  \\
				$A_2$ & \xmark & \xmark & \cmark  \\
				$A_3$ & \cmark & \cmark & \xmark  \\				
				$A_4$ & \cmark & \cmark & \cmark  	\\
				\bottomrule
			\end{tabular}
		\end{center}
		First, note that the $\fairshare(1)=\fairshare(2)=\nicefrac 7 3$ and $\fairshare(3)=\fairshare(4)=2.5$.
		The budget allocation $\{p_1,p_2\}$ satisfies FS-1, since adding project $p_3$. gives each agent at least their fair share.
		The budget allocation $\{p_3\}$ does not satisfy FS-1, as voters~$3$ and~$4$ do not achieve their fair share by adding only $p_1$ or $p_2$. It does, however, satisfy Local-FS, since adding $p_1$ or $p_2$ leads to too large a share for voter~$4$.
		Concretely, $\share(\{p_1,p_3\}, 4)=2+\nicefrac 7 3>2.5$ and $\share(\{p_2,p_3\}, 4)=\nicefrac 3 2+\nicefrac 7 3>2.5$.
		Consequently, Local-FS does not imply FS-1 even if an FS-1 allocation exists.
		
		Note, however, that $\{p_1,p_2\}$ satisfies both Local-FS and FS-1. So we have not ruled out the possibility that the existence of an FS-1 allocation implies the existence of an  allocation satisfying both FS-1 and Local-FS.
	\end{example}
	
	\section{Additional Details for the Experimental Analysis}
	
	In the following, we provide additional details regarding the experimental analysis.

	\subsection{Optimal Distance to Fair Share}
	
	Let us first focus on the first experiment.
	
	We first provide more information about the pre-processing step.
	To that end, Figure~\ref{fig:num_proj_preprocessed} presents the proportion of projects that have been removed by
	the different pre-processing operations.
	
	\begin{figure}
		\includegraphics[width=\linewidth]{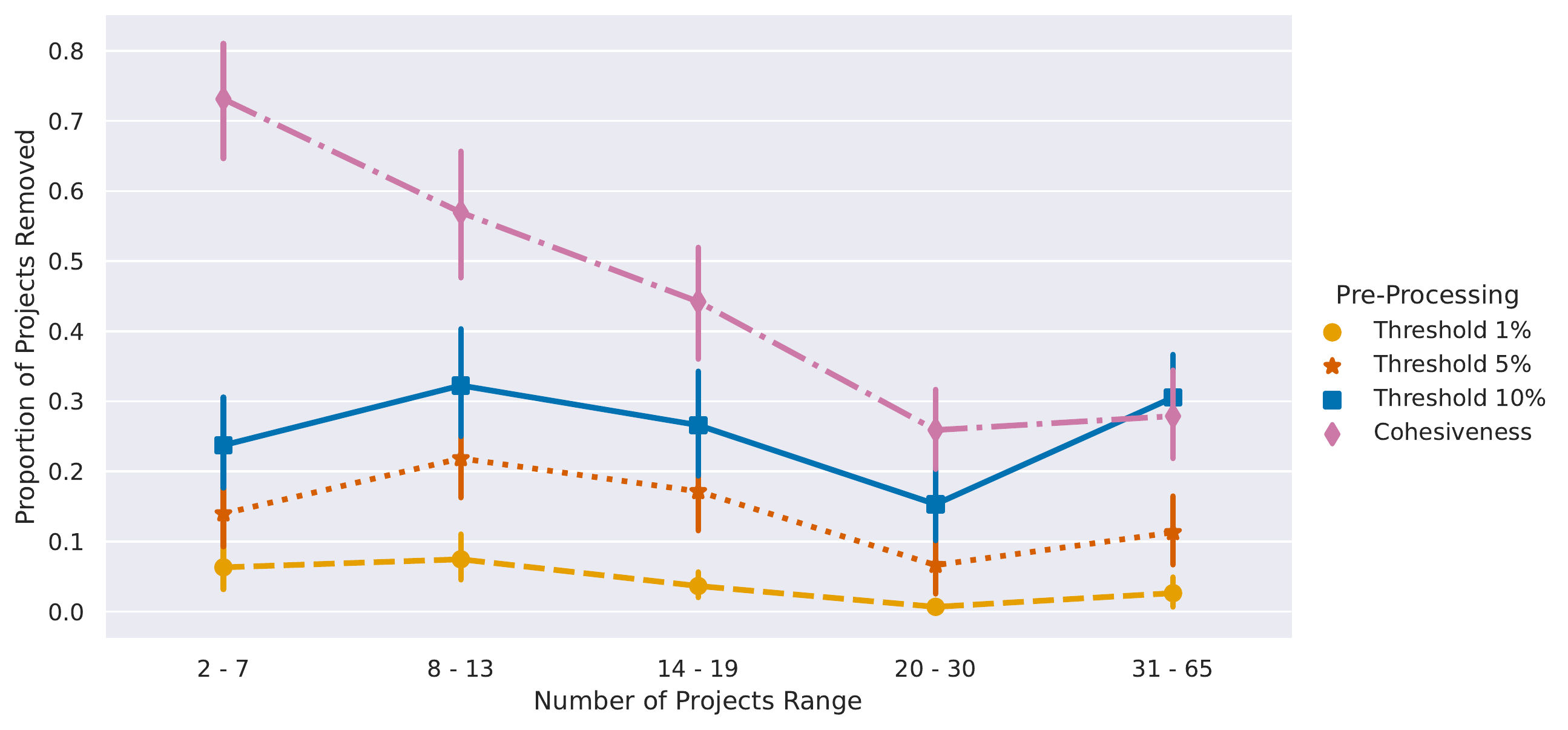}
		\caption{Proportion of projects removed during the pre-processing stage.}
		\label{fig:num_proj_preprocessed}
	\end{figure}
	
	We also document the how sensitive to the approximation ratio our results for the average capped fair share ratio
	are.
	Figure~\ref{fig:approx_ratio_capped_fs_ratio} presents the average capped fair share ratio achieved for different
	approximation ratios.
	We can see that one needs to signigicantly reduce the approximation ratio to obtain significant changes in the
	value of the average capped fair share ratio.
	
	\begin{figure}
		\includegraphics[width=\linewidth]{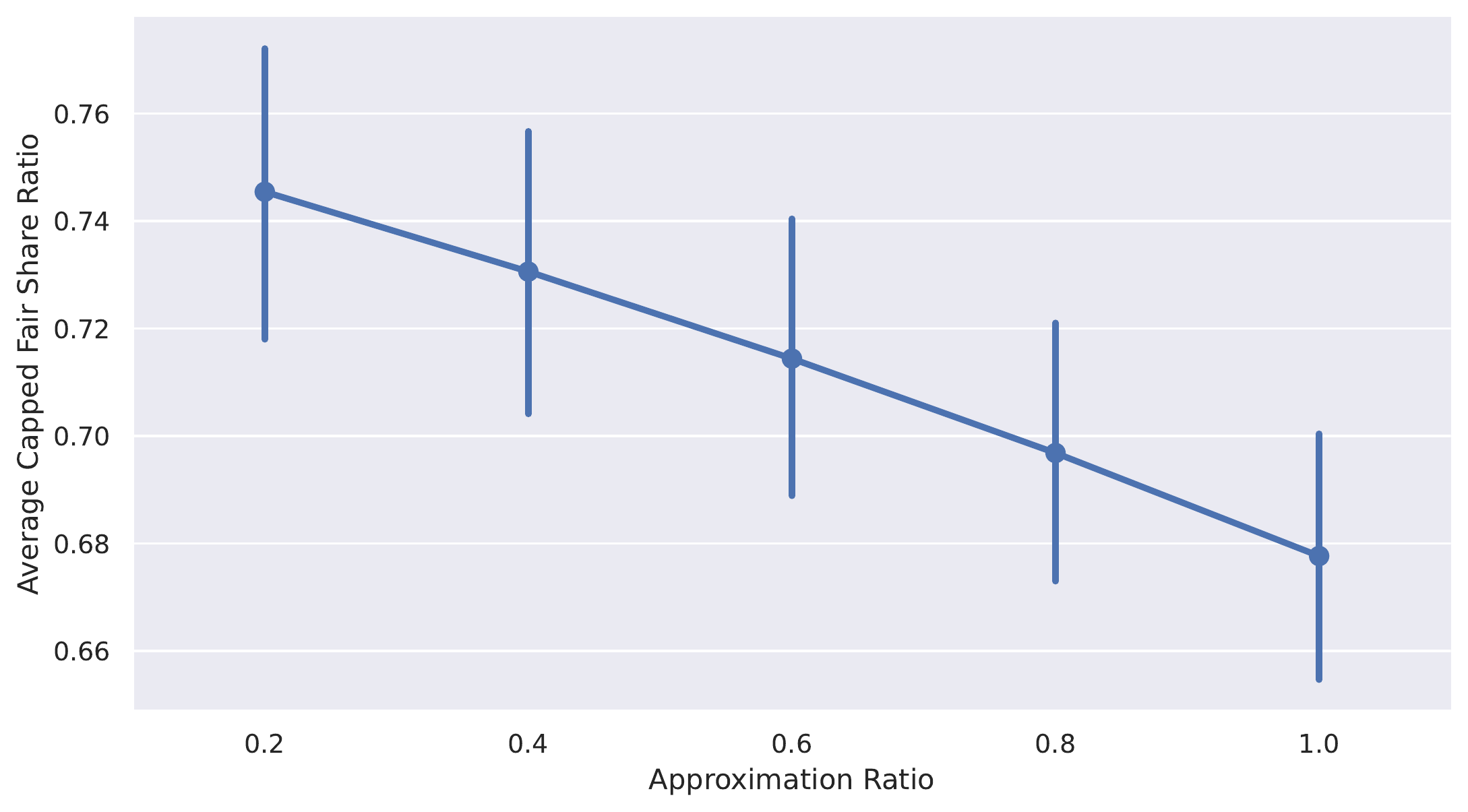}
		\caption{Average capped fair share ratio for different approximation ratio. For a given approximation ratio
			$\alpha \in (0, 1]$, the measure of interest is $\frac{1}{n} \cdot \sum_{i \in \agentSet} \min \left(
			\frac{\share(\pi, i)}{\fairshare(i) \cdot \alpha}, 1 \right)$.}
		\label{fig:approx_ratio_capped_fs_ratio}
	\end{figure}

	\subsection{Distance to Fair Share of Common PB Rules}

	For our second experiment, we first define formally all the rules that are considered.
	
	First, the two new MES rules, namely MES$_{\mathit{card}}$ and MES$_{\mathit{cost}}$, are defined similarly as
	\rulex{} but with different contribution functions.
	
	\begin{definition}[MES$_{\mathit{card}}$]
		Given an instance $I$ and a profile $\profile$, MES$_{\mathit{card}}$ constructs a budget allocation $\pi$,
		initially empty, iteratively as follows.
		A load $\ell_i: 2^\projSet \rightarrow \Rplus$, is associated with every agent $i \in \agentSet$, initialised
		as $\ell_i(\emptyset) = 0$ for all $i \in \agentSet$.
		Given $\pi$ and a scalar $\alpha \geq 0$, the card-contribution of agent $i \in \agentSet$ for project
		$p \in \projSet \setminus \pi$ is defined by:
		\[\gamma_i^{\mathit{card}}(\pi, \alpha, p) = \min\left(\nicefrac{b}{n} - \ell_i(\pi), \alpha \cdot
		\mathds{1}_{p \in A_i}\right).\]
		Given a budget allocation $\pi$, a project $p \in \projSet \setminus \pi$ is said to be
		$\alpha$-card-affordable, for $\alpha \geq 0$, if $\sum_{i \in \agentSet} \gamma_i^{\mathit{card}}
		(\pi, \alpha, p) \cdot \mathds{1}_{p \in A_i} = c(p).$
		
		At a given round with current budget allocation $\pi$, if no project is $\alpha$-card-affordable for any
		$\alpha$, MES$_{\mathit{card}}$ terminates.
		Otherwise, it selects a project $p \in \projSet \setminus \pi$ that is $\alpha^*$-card-affordable where
		$\alpha^*$ is the smallest $\alpha$ such that one project is $\alpha$-card-affordable ($\pi$ is updated to
		$\pi \cup \{p\}$). The agents' loads are then updated: If $p \notin A_i$, then $\ell_i(\pi \cup \{p\}) =
		\ell_i(\pi)$, and otherwise $\ell_i(\pi \cup \{p\}) = \ell_i(\pi) + \gamma_i^{\mathit{card}}(\pi, \alpha, p)$.
	\end{definition}

	\begin{definition}[MES$_{\mathit{cost}}$]
		Given an instance $I$ and a profile $\profile$, MES$_{\mathit{cost}}$ constructs a budget allocation $\pi$,
		initially empty, iteratively as follows.
		A load $\ell_i: 2^\projSet \rightarrow \Rplus$, is associated with every agent $i \in \agentSet$, initialised
		as $\ell_i(\emptyset) = 0$ for all $i \in \agentSet$.
		Given $\pi$ and a scalar $\alpha \geq 0$, the cost-contribution of agent $i \in \agentSet$ for project
		$p \in \projSet \setminus \pi$ is defined by:
		\[\gamma_i^{\mathit{cost}}(\pi, \alpha, p) = \min\left(\nicefrac{b}{n} - \ell_i(\pi), \alpha \cdot
		\mathds{1}_{p \in A_i} \cdot c(p)\right).\]
		Given a budget allocation $\pi$, a project $p \in \projSet \setminus \pi$ is said to be
		$\alpha$-cost-affordable, for $\alpha \geq 0$, if $\sum_{i \in \agentSet} \gamma_i^{\mathit{cost}}
		(\pi, \alpha, p) \cdot \mathds{1}_{p \in A_i} = c(p).$
		
		At a given round with current budget allocation $\pi$, if no project is $\alpha$-cost-affordable for any
		$\alpha$, MES$_{\mathit{cost}}$ terminates.
		Otherwise, it selects a project $p \in \projSet \setminus \pi$ that is $\alpha^*$-cost-affordable where $\alpha^*$ is the
		smallest $\alpha$ such that one project is $\alpha$-cost-affordable ($\pi$ is updated to $\pi \cup \{p\}$).
		The agents' loads are then updated: If $p \notin A_i$, then $\ell_i(\pi \cup \{p\}) = \ell_i(\pi)$, and
		otherwise $\ell_i(\pi \cup \{p\}) = \ell_i(\pi) + \gamma_i^{\mathit{cost}}(\pi, \alpha, p)$.
	\end{definition}
	
	\noindent We now introduce the sequential Phragmén rule.
	Note that we defined it here in its discrete formulation.
	For the continuous formulation, we refer the reader to \citet{LCG22}.
	
	\begin{definition}[Sequential Phragmén]
		Given an instance $I$ and a profile $\profile$, the sequential Phragmén rule constructs a budget allocation $\pi$,
		initially empty, iteratively as follows.
		A load $\ell_i: 2^\projSet \rightarrow \Rplus$, is associated with every agent $i \in \agentSet$, initialised
		as $\ell_i(\emptyset) = 0$ for all $i \in \agentSet$.
		Given $\pi$, the new maximum load for selecting project $p \in \projSet \setminus \pi$ is defined as:
		\[\ell^\star(\pi, p) = \frac{c(p) + \sum_{i \in \agentSet} \mathds{1}_{p \in A_i} \cdot  \ell_i(\pi)}{|\{i \in
			\agentSet \mid p \in A_i\}|}.\]
		At a given round with current budget allocation $\pi$, let $P^\star \subseteq \projSet$ be such that:
		\[P^\star = \argmin_{p \in \projSet \setminus \pi} \ell^\star(\pi, p).\]
		If there exists $p \in P^\star$ such that $c(\pi \cup \{p\}) > b$, sequential Phragmén terminates and outputs $\pi$.
		Otherwise, a project $p \in P^\star$ is selected ($\pi$ is updated to $\pi \cup \{p\}$) and the agents' load are
		updated: If $p \notin A_i$, then $\ell_i(\pi \cup \{p\}) = \ell_i(\pi)$, and otherwise $\ell_i(\pi \cup
		\{p\}) = \ell^\star(\pi, p)$.
	\end{definition}
	
	\noindent We finally define the greedy approval rule.
	
	\begin{definition}[Greedy Approval]
		Given an instance $I$ and a profile $\profile$, the greedy approval rule constructs a budget allocation $\pi$,
		initially empty, iteratively as follows.
		Projects are ordered according on their approval score (numbers of agents approving of it).
		At a given round with current budget allocation $\pi$, greedy approval consider the next project in the
		ordering, call it $p$. Project $p$ is then selected if and only if $c(\pi \cup \{p\}) \leq b$.
		The next project in the ordering, if any, is then considered.
	\end{definition}
	
	\noindent Note that for all of these rules, and \rulex{}, to return a single budget allocation, we need to use some
	tie-breaking mechanism (to distinguish between projects with same new load, same number of approvers etc\ldots).
	For our experiments, we always break ties in favour of the projects with the maximum number of approvers (and then
	lexicographically if ties persist).
	
	\begin{figure}
		\includegraphics[width=\linewidth]{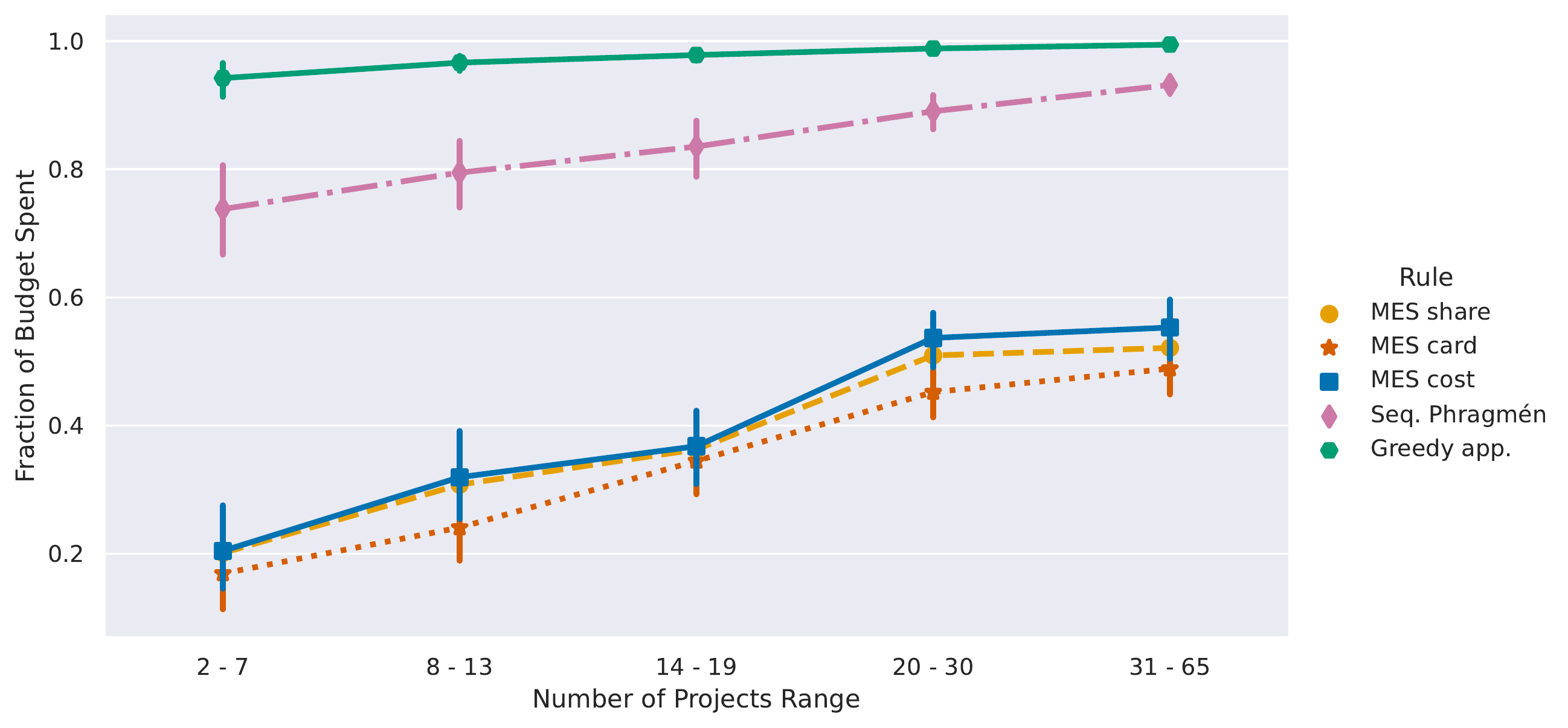}
		\caption{Fraction of the budget spent by the different rules.}
		\label{fig:fraction_spent}
	\end{figure}
	
	In the main text of the paper, we mention the difference in the budget used by the different rules.
	Figure~\ref{fig:fraction_spent} presents the cost of the outcome of the rules, expressed in fraction of the
	budget limit.
\end{document}